\documentclass[a4paper, 11pt]{article}
\usepackage[english]{babel}
\usepackage{amsmath, amssymb, amsfonts, amsthm, bm}
\usepackage{mathrsfs}
\usepackage{enumerate}
\usepackage{float}
\usepackage{color,graphicx,epstopdf}
\usepackage{dsfont}
\usepackage{hyperref}
\usepackage{cite} % must be loaded after hyperref
\usepackage{algorithm} %format of the algorithm
\usepackage{algorithmic} %format of the algorithm
\usepackage{subfigure}

\newcommand{\R}{\mathbb{R}}

\DeclareMathOperator{\diag}{diag}
\DeclareMathOperator{\rank}{rank}

\DeclareMathOperator{\colsp}{colsp}
\DeclareMathOperator{\argmin}{argmin}

\newcommand{\MO}{\mathcal{O}}
\newcommand{\sigm}{\sigma_{\rm m}}
\newcommand{\sigsecm}{\sigma_{\rm 2}}

\newcommand{\hb}{\mathbf{h}}

\newcommand{\xb}{\mathbf{x}}
\newcommand{\yb}{\mathbf{y}}
\newcommand{\zb}{\mathbf{z}}

\newcommand{\Hb}{\mathbf{H}}
\newcommand{\Lb}{\mathbf{L}}
\newcommand{\Ib}{\mathbf{I}}
\newcommand{\Mb}{\mathbf{M}}

\newcommand*\mcup{\mathbin{\mathpalette\mcupinn\relax}}
\newcommand*\mcupinn[2]{\vcenter{\hbox{$\mathsurround=0pt
  \ifx\displaystyle#1\textstyle\else#1\fi\bigcup$}}}

  \newcommand{\mU}{\mathscr{U}}

\newtheorem{theorem}{Theorem}
\newtheorem{definition}{Definition}

\newtheorem{lemma}{Lemma}

\newtheorem{assumption}{Assumption}

\title{\bf Network Flows that Solve  Least Squares \\ for Linear Equations\thanks{A preliminary version \cite{yangcdc2017} of this work was presented at the 56th IEEE Conference on Decision and Control, December 12-15, 2017 in Melbourne, Australia.}}
\date{}
\author{Yang Liu, Youcheng Lou, Brian D. O. Anderson, Guodong Shi\thanks{Y. Liu and G. Shi are with the Research School of Engineering, The Australian National University, Canberra 0200, Australia. (Email: yang.liu@anu.edu.au, guodong.shi@anu.edu.au)}%
\thanks{Y. Lou is with MDIS, Academy of Mathematics and Systems Science, Chinese Academy of Sciences, Beijing 100190, China. (Email: louyoucheng@amss.ac.cn)}
\thanks{Brian D. O. Anderson is with the Research School of Engineering, The Australian National University, Canberra 0200, Australia; Hangzhou Dianzi University, Hangzhou 310018, China; Data61-CSIRO, Canberra 0200, Australia. (Email: brian.anderson@anu.edu.au)}%%
 }

\begin{document}
\maketitle

\begin{abstract}
This paper presents a first-order {distributed continuous-time algorithm} for computing the least-squares solution to a linear equation over networks. Given the uniqueness of the solution{, with} nonintegrable {and diminishing }step size, convergence results are provided for fixed graphs. {The exact rate of convergence is also established for various types of step size choices falling into that category.} For {the case where non-unique solutions exist}, convergence {to one such solution} is proved for constantly connected switching graphs {with square integrable step size}, and for uniformly jointly connected switching graphs under the boundedness assumption on system states. Validation of the results and illustration of the impact of step size on the convergence speed are made using a few numerical examples.
\end{abstract}

\section{Introduction}\label{sec:intro}
In modern engineering systems, there is a great demand for large-scale computing capabilities for solving real-world mathematical problems. Centralized algorithms are effective tools if the computing center possesses the information of the entire problem. In some cases, however, due to the comparatively weak computing power of any one agent or its limited access to the parameters and measurement data relevant to the whole problem, the notion of distributed computation over networks has been developed \cite{tsitsiklis1984,tsitsiklis1986,lynch1996,jadbabaie2003,rabbat2005,mesbahi2010}. Nowadays it is widely applied in the areas of analyzing the consensus of complex systems \cite{olfati2004consensus}, solving various optimization problems \cite{nedic2009distributed1}, carrying out distributed estimation \cite{cattivelli2008diffusion} and filtering \cite{kar2011gossip}.

Solving systems of linear equations using distributed algorithms over networks
% {which can be a crucial requirement} in solving distributed optimization problems,
{emerges as one of the basic tasks in distributed computation.} In these scenarios, it is often assumed that each agent of the network only has access to one or a few of the individual linear equations making up the full system due to security issues or memory limitation, and is only permitted to interact with a subset of the other agents.
% Each agent of the network may request the entire solution, instead of only its components, for the purpose of protecting the privacy of customers \cite{wang2013harnessing}.
A number of contributions have been made to the development of distributed solvable linear equation solvers, where simple first-order distributed algorithms, in continuous-time or discrete-time \cite{anderson2015decentralized,shi2016network,lu2009distributed1,liu2013asynchronous,liu2014,mou2013fixed,mou2015distributed,wang2014solving,wang2016distributed}, manage to deliver satisfactory solutions even for switching network structures. As is known to all, however, another frequent case in practical problems is concerned with non-solvable linear equations, in which we often seek a least-squares solution by minimizing the associated objective function.

However, it seems a rather challenging problem in developing distributed least-squares solvers for network linear equations, due to the mismatch between individual linear equations at each node and the network least-squares solution. Despite the difficulties, there exist a few distributed algorithms developed for the least-squares problem using different approaches, such as second-order algorithms \cite{wang:elia:10,wang2012distributed,gharesifard2014distributed,yang2017exponential}, state expansion \cite{mou2015distributed} and the high gain consensus gain method \cite{shi2016network}. Second-order distributed least-squares solvers \cite{wang:elia:10,wang2012distributed,gharesifard2014distributed,yang2017exponential} generally can produce good convergence performance, however, they rely on restricted network structures and demand higher communication and storage capacities. The state expansion method \cite{mou2015distributed} is based on enlarging the state dimension and then applying the existing methods for linear equations with exact solutions directly, but a negative feature is that the nodes must have access to more knowledge than their own linear equations. It was shown in \cite{shi2016network} that first-order algorithms for exact solutions can be adapted to the least-squares case by a high consensus gain, but only in an approximate sense.

In this paper, we propose a first-order continuous-time flow for the least-squares problems of network linear equations, in which each agent keeps averaging the state with its neighbors' and at the same time descends along the negative gradient of its local cost function. This flow is inspired by the work of \cite{nedic2015distributed} on distributed subgradient optimization. If the network linear equation has one unique least-squares solution, we prove that all node states asymptotically converge to that solution along our flow, with constant and connected graphs and a step size tending to zero, but not too fast. {We also give analytical results on how the choice of step size, the attributes of linear equations and network size affect the convergence speed.} For a switching network structure that is at all times connected, we show that the node states always converge to one of the least-squares solutions with square integrable step size. The same convergence result is shown to hold for a uniformly jointly connected switching network under a boundedness assumption on the system states. We also provide a few numerical examples that validate the usefulness of the proposed algorithms and demonstrate the convergence rate.

{A preliminary version of this work \cite{yangcdc2017} was presented at the 56th IEEE Conference on Decision and Control. Compared to the conference version, we make additional contributions as follows: (i) analytical  studies on the rate of convergence of the proposed algorithm are provided; 
(ii) convergence results are  stated under a common structure for all network and linear equation scenarios, in addition to the detailed proofs; (iii) more numerical validations are presented. The remainder of this paper is organized as follows. In Section \ref{sec:pre}, a brief introduction to the definition of the problem studied is given. We present the main results in Section \ref{sec:results} and provide their detailed proofs in Section \ref{sec:proofs}. We also provide validations and further discussions using numerical examples in Section \ref{sec:sims}. In Section \ref{sec:conclu}, the main work of this paper is summarized and potential future work directions are provided.

\section{Problem Definition}\label{sec:pre}
In this section, a few mathematical preliminaries are provided, regarding linear equations over networks. Also we establish a distributed network flow that can asymptotically compute the least-squares solution to network linear equations and discuss its relation to existing work.
\subsection{Linear Equations}
Consider the following linear algebraic equation with respect to $\yb\in \R^m$
\begin{equation} \label{eq:linear_equation}
\zb = \Hb \yb{,}
\end{equation}
where $\zb\in\R^N$ and $\Hb\in\R^{N\times m}$ are known and satisfy $N\ge m$. Denote
\[
\Hb =
\begin{bmatrix}
\hb_1^\top \\
\hb_2^\top \\
\vdots \\
\hb_N^\top
\end{bmatrix},
\quad
\zb =
\begin{bmatrix}
z_1 \\
z_2 \\
\vdots \\
z_N
\end{bmatrix}
\]
with {$\hb_i\in\R^m$ for all $i=1,\dots,N$}. We can rewrite (\ref{eq:linear_equation}) as
\[
\hb_i^\top \yb = z_i,\ i=1,\dots,N.
\]
Denote the column space of a matrix $\Mb$ by $\colsp\{\Mb\}$. If $\zb\in\colsp\{\Hb\}$, then the equation (\ref{eq:linear_equation}) always has (one or many) exact solutions. If $\zb\notin\colsp\{\Hb\}$, the least-squares solution is defined by the solution of the following optimization problem:
\begin{equation}\label{eq:leastSquare_constaint}
\min_{\yb\in\R^m} \|\zb-\Hb\yb\|^2.
\end{equation}

It is well known that if $\rank(\Hb)=m$, then (\ref{eq:leastSquare_constaint}) yields a unique solution $\yb^\ast=(\Hb^\top \Hb)^{-1}\Hb^\top\zb$, while (\ref{eq:leastSquare_constaint}) has a set of non-unique least-squares solutions if $\rank(\Hb)<m$. Define
\begin{equation}\notag
f(\yb)=\|\zb-\Hb\yb\|^2=\sum\limits_{i=1}^{N}f_i(\yb),
\end{equation}
where $f_i(\yb)=|\hb_i^\top\yb-z_i|^2$. Note that $\yb^\ast\in\argmin f(\yb)$, i.e., $\nabla f(\yb^\ast)=0$, where $\nabla f(\yb) = 2\sum\limits_{i=1}^{N}(\hb_i\hb_i^\top\yb-z_i\hb_i)$.
% \begin{equation}\notag
% \begin{aligned}
% \nabla f(\yb) &= 2\sum\limits_{i=1}^{N}(\hb_i\hb_i^\top\yb-z_i\hb_i),\\
% \nabla^2 f(\yb) &= 2\sum\limits_{i=1}^{N}\hb_i\hb_i^\top,
% \end{aligned}
% \end{equation}
% and $\yb^\ast=\argmin f(\yb)$, i.e., $\nabla f(\yb^\ast)=0$.
% We can also represent (\ref{eq:leastSquare_constaint}) by the cost function as the following:
% \begin{equation}\notag
% \min_{\yb\in\R^m} f(\yb).
% \end{equation}

\subsection{Networks}
Let $\mathcal{G}=(\mathcal{V},\mathcal{E})$ denote a constant, undirected and simple graph with the finite set of nodes $\mathcal{V}=\{1,2,\dots,N\}$ and the set of edges $\mathcal{E}=\big\{\{i,j\}:i,j\in\mathcal{V}\textnormal{ are connected}\big\}$. {Let $\R^+,\R^{\ge 0}$ denote the sets of all positive real numbers and nonnegative real numbers, respectively.} Define a weight function $w:\mathcal{E}\to\R{^+}$ over the edge set with the weight of edge $\{i,j\}$ being $w(\{i,j\})$. {It is worth noting the weight $w$ for each edge is assumed to be fixed in this paper for ease of the presentation. Generalizations to time-varying weights can be made similarly to the analysis of \cite{shi2016network}.} Based on constant graphs, we next introduce time-varying graphs. Let $\mathcal{Q}$ be the set containing all possible constant and undirected graphs induced by the node set $\mathcal{V}$ and let $\mathcal{Q}^\ast\subset\mathcal{Q}$ be a subset of $\mathcal{Q}$. Define a piecewise constant mapping $\mathcal{G}_{\sigma}=(\mathcal{V},\mathcal{E}_{\sigma}):\R^{\ge 0}\to\mathcal{Q}^\ast$. {Throughout this paper, we assume the set of times corresponding to discontinuities of $\mathcal{G}_{\sigma(t)}$ has measure zero.}
% Also define $\tau_d$ as the dwell time if there exists $\tau_d>0$ such that all intervals between consecutive discontinuities are no small than $\tau_d$.
Note that the time-varying graph $\mathcal{G}_{\sigma(t)}=(\mathcal{V},\mathcal{E}_{\sigma(t)})$ represents the network topology at time $t$. Let $\mathcal{N}_i(t)$ be the set of neighbor nodes that are connected to node $i$ at time $t$, i.e., $\mathcal{N}_i(t)=\big\{j:\{i,j\}\in\mathcal{E}_{\sigma(t)}\big\}$. Define the adjacency matrix $\mathbf{A}(t)$ of the graph $\mathcal{G}_{\sigma(t)}$ by $[\mathbf{A}(t)]_{ij}=w(\{i,j\})$ if $\{i,j\}\in\mathcal{E}_{\sigma(t)}\}$, and $[\mathbf{A}(t)]_{ij}=0$ otherwise, and $\mathbf{D}(t)=\diag(\sum\limits_{j=1}^N[\mathbf{A}(t)]_{1j},\dots,\sum\limits_{j=1}^N[\mathbf{A}(t)]_{Nj})$. Then $\mathbf{L}(t)=\mathbf{D}(t)-\mathbf{A}(t)$ is the Laplacian of graph $\mathcal{G}_{\sigma(t)}$ at time $t$.

% In the following we provide an essential definition.
% \begin{definition}
% Consider a time-varying graph $\mathcal{G}_{\sigma(t)}=(\mathcal{V}, \mathcal{E}_{\sigma(t)})$. The joint graph of $\mathcal{G}_{\sigma(t)}$ in the time interval $[t_1,t_2)$ with $t_1<t_2\le\infty$ is denoted as $\mathcal{G}([t_1,t_2))=(\mathcal{V},\cup_{t\in[t_1,t_2)}\mathcal{E}_{\sigma(t)})$. Then $\mathcal{G}_{\sigma(t)}$ is uniformly jointly connected if there exists a constant $T>0$ such that $\mathcal{G}([t,t+T))$ is connected for all $t\ge 0$.
% \end{definition}

\subsection{Distributed Flows}
Assume that node $i$ of the network $\mathcal{G}_{\sigma(t)}$ only knows the information of $\hb_i,z_i$, i.e., node $i$ is associated with the linear equation $\hb_i^\top\yb=z_i$. We associate with each node $i$ a state $\xb_i(t)\in\R^m$, which{, as the notation implies,} in general varies with time. Then we propose the following continuous-time network flow
\begin{equation}\label{fixed}
  \dot{\xb}_i(t) = K\sum\limits_{j\in\mathcal{N}_i(t)}[\mathbf{A}(t)]_{ij}(\xb_j(t)-\xb_i(t)) - \frac{\alpha(t)}{2}\nabla f_i(\xb_i(t)),
\end{equation}
where $K\in\R^+$ is a positive constant, $\nabla f_i(\yb)=2(\hb_i\hb_i^\top\yb-z_i\hb_i)$ and the step size $\alpha:\R^{\geq 0}\to\R^+$ is a continuous function which assures the continuity of all $\xb_i(t)$ and their derivatives, with the exception of the time points when the networks switch. In vector form, we have
\begin{equation}\label{eq:x_ori}
\dot{\xb}(t) = -\Mb(t)\xb(t)+\alpha(t)\zb_H,
\end{equation}
where
{\begin{equation}\notag
\begin{aligned}
\xb(t) &= \begin{bmatrix}
\xb_1(t)^\top & \dots & \xb_N(t)^\top
\end{bmatrix}^\top,\\
\Mb(t) &= K(\Lb(t)\otimes\Ib_m)+\alpha(t)\tilde{\Hb},\\
\tilde{\Hb} &= \diag\big(\hb_1\hb_1^\top,\dots,\hb_N\hb_N^\top\big),\\
\zb_H &= \begin{bmatrix}
z_1\hb_1^\top & \cdots & z_N\hb_N^\top
\end{bmatrix}^\top.
\end{aligned}
\end{equation}}
Now we make several assumptions of $\alpha(t)$ that will be used in our main results.
\begin{assumption}\label{ass:1}
(i) $\int_{0}^\infty \alpha(t)dt =\infty$; (ii) $\lim\limits_{t\to\infty}\alpha(t)=0$; (iii) $\int_{0}^\infty \alpha^2(t)dt <\infty$.
\end{assumption}
% {It is easy to verify that the functions in the set $\{\alpha:\R^{\ge 0}\to\R^+:\alpha(t)=(t+c)^{b},-1\le b<0,c\in\R\}$ and their linear combinations satisfy Assumption \ref{ass:1} 1 (i) (ii) and further the functions in $\{\alpha:\R^{\ge 0}\to\R^+:\alpha(t)=(t+c)^{b},-1\le b\le-\frac{1}{2},c\in\R\}$  and their linear combinations satisfy all requirements of Assumption 1.}

\subsection{Discussion}
Now we clarify the relation between the previous work on distributed least-squares and optimization algorithms, and our algorithm (\ref{fixed}) by briefly discussing their structure and applicability. {}It is clear that (\ref{fixed}) has exactly the same structure as the flow in \cite{nedic2010constrained,nedic2015distributed} {in the sense that they are both in the form of ``local averaging consensus" + ``diminishing local objective"}, with the difference that the flow in \cite{nedic2010constrained,nedic2015distributed} is discrete-time but (\ref{fixed}) is continuous-time. However, we cannot use the algorithm and the analysis directly because the gradient boundedness of (\ref{fixed}) is not directly verifiable. It can be noted that the first-order flow in \cite{shi2016network} is a special case of (\ref{fixed}) obtained by letting $\alpha(t)$ be some constant. Due to the existence of the diminishing step size, (\ref{fixed}) is a linear time-varying system, while the flow in \cite{shi2016network} is linear time-invariant and can only produce the solution in approximate sense. Hence the approach to analyzing the flow in \cite{shi2016network} is not applicable for (\ref{fixed}). {Indeed (\ref{fixed}) can be formulated by properly specializing the optimization problem in \cite{touri2015} and letting each agent's output scale be constant one. However, because of the specificity of the least-squares cost function, relaxed convergence conditions become possible as will be shown later.} {In addition, we will provide analytic results on the convergence speed for the fixed network case.} There are also second-order least-squares solvers \cite{wang:elia:10,wang2012distributed,gharesifard2014distributed,yang2017exponential}, but they often require limited network topologies and have more complex structures than (\ref{fixed}).

\section{Main Results}\label{sec:results}
In this section, we investigate the flow (\ref{eq:x_ori}) over fixed and switching networks, respectively, and establish the convergence conditions regarding $\alpha(t)$ and the graphs.\\
\noindent Proofs of the results appear in later subsections.

\subsection{Convergence over Fixed Networks}
First we consider the case where the linear equation (\ref{eq:linear_equation}) has one unique least-squares solution and the network is a constant graph for all $t$. In this case, the following theorem holds.

\begin{theorem}\label{thm:1} Let $\yb^\ast=(\Hb^\top \Hb)^{-1}\Hb^\top\zb$ denote the unique least-squares solution of (\ref{eq:linear_equation}) and suppose $\rank(\Hb)=m$. Let Assumption \ref{ass:1} (i) and (ii) hold. If $\mathcal{G}_{\sigma(t)}=\mathcal{G}$ is constant and connected for all $t\ge 0$, then along any solution of (\ref{fixed}) there holds
\[
\lim_{t\to \infty} \xb_i(t)=\yb^\ast
\]
for all $i\in\mathcal{V}$.
% \begin{equation}\label{eq:convergence_rate}
% \big\|\sum\limits_{i=1}^N \xb_i(t)/N-\hat{\yb}\big\|^2\in\left\{
% \begin{aligned}
% & \MO(\frac{1}{t^{\min(\frac{1}{2},\frac{\sigma}{N})}})&\textnormal{ if }\alpha(t)=\MO(\frac{1}{t});\\
% & \MO(\frac{1}{\sqrt{t^{3\lambda}}})&\textnormal{ if }\alpha(t)=\MO(\frac{1}{t^\lambda}),
% \end{aligned}
% \right.
% \end{equation}
% \begin{equation}\label{eq:convergence_rate}
% \big\|\sum\limits_{i=1}^N \xb_i(t)/N-\hat{\yb}\big\|^2\in\left\{
% \begin{aligned}
% & \MO(\frac{1}{t^{\min(\frac{1}{2},\frac{\sigma}{N})}})&\textnormal{ if }\alpha(t)=\MO(\frac{1}{t});\\
% & \MO(\frac{1}{\sqrt{t^{3\lambda}}})&\textnormal{ if }\alpha(t)=\MO(\frac{1}{t^\lambda}),
% \end{aligned}
% \right.
% \end{equation}
% where $\sigma>0$ is the smallest eigenvalue of $\Hb^\top\Hb$, $\lambda\in(0,1)$.
\end{theorem}

{Let $\sigm(\cdot)$ and $\sigsecm(\cdot)$ denote the smallest and the second smallest eigenvalue of a real symmetric matrix, respectively. For two functions $g,h:\R^{\ge0}\to\R^+$, we say $g(t)=\MO(h(t))$ if there exist $c>0$ and $\tau>0$ such that $g(t) \le c\cdot h(t)$ for all $t\ge\tau$.
% {Define an indicator $\mathcal{I}:\R^+\to\{0,1\}$ with
% \begin{equation}\notag
% \mathcal{I}(x)=\left\{
% \begin{aligned}
% &  1&\textnormal{ if }x=1;\\
% & 0&\textnormal{ otherwise}.
% \end{aligned}
% \right.
% \end{equation}}

The following theorem characterizes the convergence speed of the algorithm (\ref{fixed}) for different choices of step size {known to decay with a $t$'s inverse power that is no bigger than one.}
\begin{theorem}\label{thm:convergence_rate}
Suppose the conditions of Theorem \ref{thm:1} hold. Define $\yb^\ast=(\Hb^\top \Hb)^{-1}\Hb^\top\zb$.
\begin{enumerate}[(i)]
\item If $\alpha(t)=\MO({\frac{1}{t}})$, then along (\ref{fixed}) there hold
\begin{enumerate}[(a)]
\item $\big\|\sum\limits_{i=1}^N \xb_i(t)/N-\yb^\ast\big\| = \MO\bigg(\frac{1}{t^{\min(1,\frac{\sigm(\Hb^\top\Hb)}{N})}}\bigg)\quad$ for $\ \sigm(\Hb^\top\Hb)\neq N$.
\item $\big\|\sum\limits_{i=1}^N \xb_i(t)/N-\yb^\ast\big\| = \MO\big(\frac{\log t}{t}\big)\qquad\qquad\qquad $ for $\ \sigm(\Hb^\top\Hb)=N$.
\end{enumerate}
\medskip
% \begin{equation}\label{eq:convergence_rate}
% \begin{aligned}
% &\quad\big\|\sum\limits_{i=1}^N \xb_i(t)/N-\yb^\ast\big\|\\
% &=\left\{
% \begin{aligned}
% &\MO\bigg(\frac{\log t}{t}\bigg)\qquad\qquad\quad\textnormal{ if }\frac{\sigm(\Hb^\top\Hb)}{N}=1\\
% &\MO\bigg(\frac{1}{t^{\min(1,\frac{\sigm(\Hb^\top\Hb)}{N})}}\bigg)\ \textnormal{ otherwise}.
% \end{aligned}
% \right.
% \end{aligned}
% \end{equation}}
\item If $\alpha(t)=\MO({\frac{1}{t^\lambda}})$ for $\lambda\in(0,1)$, then along (\ref{fixed}) there holds
\begin{equation}\notag
\big\|\sum\limits_{i=1}^N \xb_i(t)/N-\yb^\ast\big\| = \MO\bigg(\frac{1}{t^{\lambda}}\bigg).
\end{equation}
\end{enumerate}
\end{theorem}
Clearly, Theorem \ref{thm:convergence_rate} provides some guidance on the choice of the step size $\alpha(t)$ to guarantee fast convergence speed as follows:
\begin{enumerate}[(i)]
\item For linear equations and networks with {$\frac{\sigm(\Hb^\top\Hb)}{N}\ge 1$}, $\alpha(t)=\MO(\frac{1}{t})$ yields the fastest convergence speed.
\item For linear equations and networks with $\frac{\sigm(\Hb^\top\Hb)}{N}<1$, $\alpha(t)=\MO(\frac{1}{t^\lambda})$ with $\frac{\sigm(\Hb^\top\Hb)}{N}<\lambda<1$ admits the fastest convergence speed. In this case, the rate of convergence will increase as $\lambda$ becomes larger. Interestingly however, when $\lambda$ reaches one, the rate of convergence suddenly drops to that of the case $\lambda=\frac{\sigm(\Hb^\top\Hb)}{N}$.
\end{enumerate}
These results, especially the discontinuity around the inverse power one of $t$, would have been difficult to predict. As will be shown later, numerical results demonstrate that the convergence upper bounds established in Theorem 3 are also the asymptotic lower bounds.

% \begin{enumerate}[(i)]
% \item For large-scale networks, $\frac{2\sigm(\Hb^\top\Hb)}{N}$ is generally small. In this case, $\alpha(t)=\frac{1}{t^\lambda}$ with $\frac{\sigm(\Hb^\top\Hb)}{N}<\lambda<1$ admits fast convergence speed. The larger $\lambda$ is, the faster the network states converge.
% \item For small-scale networks, $\frac{2\sigm(\Hb^\top\Hb)}{N}$ can be large. Fast convergence speed is guaranteed with the choice $\alpha(t)=\frac{1}{t}$.
% \end{enumerate}}

\subsection{Convergence over Switching Networks}
Now we consider a more general case where the least-squares solutions of (\ref{eq:linear_equation}) can be unique or non-unique, and the network $\mathcal{G}_{\sigma(t)}$ switches among a collection of graphs. {Evidently, the Caratheodory solutions of (\ref{eq:x_ori}) exist for all initial conditions because the set of times corresponding to discontinuities of $\mathcal{G}_{\sigma(t)}$ is assumed to have measure zero.}
% Regarding the convergence of (\ref{fixed}) in this case, we have the following theorem.

\begin{theorem}\label{thm:2}Suppose $\rank(\Hb)\le m$ and denote the set of least-squares solutions of (\ref{eq:linear_equation}) by $\mathcal{Y}_{\rm LS}=\argmin f(\yb)$. In particular, $|\mathcal{Y}_{\rm LS}|=1$ if $\rank(\Hb)=m$. Suppose Assumption \ref{ass:1} (i), (ii) and (iii) hold.  If all $\mathcal{G}\in\mathcal{Q}^\ast$ are connected, then along any solution of (\ref{fixed}) over the switching graph $\mathcal{G}_{\sigma(t)}$ there exists $\hat{\yb}\in\mathcal{Y}_{\rm LS}$ such that
\[
\lim_{t\to \infty} \xb_i(t)=\hat{\yb}
\]
for all $i\in\mathcal{V}$.
\end{theorem}

In the following theorem, we prove that the connectedness condition for graphs in Theorem \ref{thm:2} can be relaxed. We provide an essential definition.
\begin{definition}
Consider a graph $\mathcal{G}_{\sigma(t)}=(\mathcal{V}, \mathcal{E}_{\sigma(t)})$. The joint graph of $\mathcal{G}_{\sigma(t)}$ in the time interval $[t_1,t_2)$ with $t_1<t_2\le\infty$ is denoted as
\begin{equation}\notag
\mathcal{G}([t_1,t_2))=(\mathcal{V},\cup_{t\in[t_1,t_2)}\mathcal{E}_{\sigma(t)}).
\end{equation}
Then $\mathcal{G}_{\sigma(t)}$ is uniformly jointly connected if there exists a constant $T>0$ such that $\mathcal{G}([t,t+T))$ is connected for all $t\ge 0$.
\end{definition}
{Let $\tau_1,\tau_2,\dots$ with $0<\tau_1<\tau_2<\dots$ denote the consecutive discontinuities of $\mathcal{G}_{\sigma(t)}$. Then we present the following assumption.
% Also define $\tau_d$ as the dwell time if there exists $\tau_d>0$ such that all intervals between consecutive discontinuities are no small than $\tau_d$.
\begin{assumption}\label{ass:dwelltime}
There exists $\tau_d>0$ such that
\begin{equation}\notag
\tau_{i+1}-\tau_i>\tau_d
\end{equation}
for all $i=0,1,2,\dots$ where $\tau_0=0$.
\end{assumption}}
Then we have the following result.
\begin{theorem}\label{thm:3}
Let $\mathcal{Y}_{\rm LS}=\argmin f(\yb)$ be the set of least-squares solutions of (\ref{eq:linear_equation}) and suppose $\rank(\Hb)\le m$. Let Assumption \ref{ass:1} (i), (ii), (iii) { and Assumption \ref{ass:dwelltime}} hold. Suppose there exists $M>0$ such that $\|\xb(t)\|\le M$ for all $t\ge 0$. If $\mathcal{G}_{\sigma(t)}$ is uniformly jointly connected, then along any solution of (\ref{fixed}) over the switching graph $\mathcal{G}_{\sigma(t)}$ there exists $\hat{\yb}\in\mathcal{Y}_{\rm LS}$ such that
\[
\lim_{t\to \infty} \xb_i(t)=\hat{\yb}
\]
for all $i\in\mathcal{V}$.
\end{theorem}
{We must mention that it is hard to provide the conditions for which the system state $\xb(t)$ is bounded in Theorem \ref{thm:3}. However, numerical examples can show the boundedness condition is satisfied in many circumstances.}

% Now we define the graph union of $\mathcal{Q}^\ast$ as $\mathcal{G}(\mathcal{Q}^\ast)=\bigcup\limits_{\mathcal{G}\in\mathcal{Q}^\ast}\mathcal{E}$ with $\mathcal{G}=(\mathcal{V},\mathcal{E})$. One of the potential generalizations is to extend ``all graphs $\mathcal{G}\in\mathcal{Q}^\ast$ are connected" to ``the graph union $\mathcal{G}(\mathcal{Q}^\ast)$ is connected", for which we provide a numerical example in this paper.

\section{Proofs of Statements}\label{sec:proofs}
Now we provide the proofs of our main results, in addition to a couple of key lemmas.
\subsection{Key Lemmas}
We begin with several lemmas that assist with the proofs of Theorem \ref{thm:1}, Theorem \ref{thm:2} and Theorem \ref{thm:3}. {Let $\langle\cdot,\cdot\rangle$ denote the inner product of two vectors of the same dimension.} {We say a differentiable function $g:\R^N\to\R$ is $\theta$-strongly convex if
$$g(\yb_1)-g(\yb_2)\ge\nabla g(\yb_2)^\top(\yb_1-\yb_2)+\frac{\theta}{2}\|\yb_1-\yb_2\|^2$$
for all $\yb_1,\yb_2\in\R^N$.}

% \begin{lemma}\label{lem1}
% Consider a linear equation $\Hb\yb=\zb$ with respect to $\yb\in\R^m$ where $\Hb\in\R^{N\times m},\zb\in\R^N$. Let $\hb_i^\top$ denote the $i$-th row of $\Hb$ and $z_i$ denote the $i$-th entry of $\zb$. Let $f(\yb)=\sum\limits_{i=1}^{N}|\hb_i^\top\yb-z_i|^2$. If $\rank(\Hb)=m$, then {$f$ is strongly convex}. Specifically, $f(\xb){=} f(\yb)+\langle\nabla f(\yb),\xb-\yb\rangle+\sigm(\Hb^\top\Hb)\|\xb-\yb\|^2$, where $\sigm(\Hb^\top\Hb)$ is the smallest eigenvalue of the positive definite matrix $\Hb^\top\Hb$.
% \end{lemma}
% \begin{proof}
% It is clear that $\nabla^2 f(\yb)=2\sum\limits_{i=1}^N \hb_i\hb_i^\top$ is positive-definite because $\mathbf{v}^\top\nabla^2 f(\yb)\mathbf{v}=2\sum\limits_{i=1}^{N}\|\hb_i^\top\mathbf{v}\|^2>0$ for any $\vb\neq 0$ if $\rank(\Hb)=m$.

% Since by Theorem 4.2.2 in \cite{horn2012matrix} $\langle\nabla f(\xb)-\nabla f(\yb),\xb-\yb\rangle=(\xb-\yb)^\top(2\sum\limits_{i=1}^{N}\hb_i\hb_i^\top)(\xb-\yb)\ge 2\sigm(\Hb^\top\Hb)\|\xb-\yb\|^2$. By Lemma 3 in \cite[pp. 9--10]{polyak1987introduction}, Lemma \ref{lem1} can be concluded.
% \end{proof}

{
\begin{lemma}\label{lem1}
Consider a matrix $\Hb\in\R^{N\times m}$ with $N\ge m$ and a vector $\zb\in\R^N$. Define $f(\yb)=\|\Hb\yb-\zb\|^2$. If $\rank(\Hb)=m$, then $f$ is $2\sigm(\Hb^\top\Hb)$-strongly convex.
\end{lemma}
\begin{proof}
Evidently, $\Hb^\top\Hb-\sigm(\Hb^\top\Hb)\Ib$ is a positive semidefinite matrix.
Let $\yb_1,\yb_2\in\R^m$. By applying Taylor series expansion on $f$ around $\yb_2$. we obtain
\begin{equation}\notag
\begin{aligned}
&\quad f(\yb_1) - f(\yb_2) \\
&=\nabla f(\yb_2)^\top (\yb_1-\yb_2) + \frac{1}{2} (\yb_1-\yb_2)^\top\nabla^2 f(\yb_2) (\yb_1-\yb_2)\\
&=\nabla f(\yb_2)^\top (\yb_1-\yb_2) + (\yb_1-\yb_2)^\top \Hb^\top\Hb (\yb_1-\yb_2)\\
&\ge \nabla f(\yb_2)^\top (\yb_1-\yb_2) + \sigm(\Hb^\top\Hb) \|\yb_1-\yb_2\|^2,
\end{aligned}
\end{equation}
which completes the proof. \hfill $\square$
% $\nabla^2 f(\yb)=2\Hb^\top\Hb$ is positive definite, because $\mathbf{v}^\top\nabla^2 f(\yb)\mathbf{v}=2\|\Hb\vb\|^2>0$ for any $\vb\neq 0$.
% It is clear that $\nabla^2 f(\yb)=2\Hb^\top\Hb$ is positive definite, because $\mathbf{v}^\top\nabla^2 f(\yb)\mathbf{v}=2\|\Hb\vb\|^2>0$ for any $\vb\neq 0$ if $\rank(\Hb)=m$.
% Since by Theorem 4.2.2 in \cite{horn2012matrix} $\langle\nabla f(\xb)-\nabla f(\yb),\xb-\yb\rangle=(\xb-\yb)^\top(2\sum\limits_{i=1}^{N}\hb_i\hb_i^\top)(\xb-\yb)\ge 2\sigm(\Hb^\top\Hb)\|\xb-\yb\|^2$. By Lemma 3 in \cite[pp. 9--10]{polyak1987introduction}, Lemma \ref{lem1} can be concluded.
\end{proof}
}

{
\begin{lemma}\label{lem:int_bound}
Let $\mu,\lambda>0$. Then
$$\int_0^t \MO\bigg(\frac{e^{\mu s}}{s^\lambda}\bigg)\mathrm{ds} = \MO\bigg(\frac{e^{\mu t}}{t^\lambda}\bigg).$$
\end{lemma}
\begin{proof}
Introduce $\phi\in(0,\mu)$ and define $\tau=\frac{\lambda}{\mu-\phi}$. Then it can be easily shown for $t>\tau$, there holds
\begin{align}
\int_{0^+}^t \frac{e^{\mu s}}{s^\lambda}\mathrm{ds} &= \int_{0^+}^{\tau} \frac{e^{\mu s}}{s^\lambda}\mathrm{ds} + \int_{\tau}^{t} \frac{e^{\mu s}}{s^\lambda}\mathrm{ds} \notag\\
&\le \int_{0^+}^{\tau} \frac{e^{\mu s}}{s^\lambda}\mathrm{ds} + \int_{\tau}^{t} \bigg(\frac{1}{\phi}\big(\mu-\frac{\lambda}{s}\big)\bigg)\frac{e^{\mu s}}{s^\lambda}\mathrm{ds} \notag\\
&\le \int_{0^+}^{\tau} \frac{e^{\mu s}}{s^\lambda}\mathrm{ds} + \int_{\tau}^{t} \frac{\mathrm{d}}{\mathrm{ds}}\frac{e^{\mu s}}{\phi s^\lambda}\notag\\
&= \frac{e^{\mu t}}{\phi t^\lambda} + \int_{0^+}^{\tau} \frac{e^{\mu s}}{s^\lambda}\mathrm{ds} - \frac{e^{\mu \tau}}{\phi \tau^\lambda},\notag
\end{align}
which completes the proof noting the definition of $\MO(\cdot)$.

% Let $\psi\ge\mu$. Similarly , one can show for all $t>0$
% \begin{equation}\label{eq:exp_ratio_lowerbound}
% \int_{0^+}^t \frac{e^{\mu s}}{s^\lambda}\mathrm{ds} \ge \int_{0^+}^t \frac{\mathrm{d}}{\mathrm{ds}}\frac{e^{\mu s}}{\psi s^\lambda}\mathrm{ds}.
% \end{equation}
% The proof is completed by (\ref{eq:exp_ratio_upperbound}) and (\ref{eq:exp_ratio_lowerbound}).
\end{proof}
}

\begin{lemma}\label{lem2}
Consider a continuously differentiable function $g:\R^{\ge 0}\to\R^{\ge 0}$. If there exist continuous functions $\gamma:\R^{\ge 0}\to\R^+$ and $\beta:\R^{\ge 0}\to\R^+$ satisfying $\dot{g}(t)\le -\gamma(t)g(t)+\beta(t)$, then
\begin{equation}\notag
g(t)\le e^{-\int_{0}^{t}\gamma(s)\mathrm{ds}}g(0)+\int_{0}^{t}e^{-\int_s^t \gamma(r)\mathrm{dr}}\beta(s)\mathrm{ds}.
\end{equation}
Furthermore, the following statements hold:
\begin{enumerate}[(i)]
\item If $\int_{0}^\infty \gamma(t)\mathrm{dt} =\infty$ and $\lim\limits_{t\to\infty}\frac{\beta(t)}{\gamma(t)}=0$, then $\lim\limits_{t\to\infty}g(t)=0$.
\item If $\int_{0}^\infty \gamma(t)\mathrm{dt} =\infty$ and $\limsup\limits_{t\to\infty}\frac{\beta(t)}{\gamma(t)}<\infty$, then $\{g(t)\}_{t\geq0}$ is bounded.
\end{enumerate}
\end{lemma}
\begin{proof}
{The proof of the inequality of $g(t)$ follows from Gr\"{o}nwall's Inequality \cite{gronwall}.}
% Consider $\dot{g_1}(t) = -\gamma(t)g_1(t)+\beta(t)$ with $g_1(0)=g(0)$. Define $v(t)=g_1(t)\exp(\int_0^t\gamma(s)ds)$. Then
% \begin{equation}\notag
% \begin{aligned}
% \dot{v}(t) & = \dot{g_1}(t)\exp\big(\int_0^t\gamma(s)ds\big) + \gamma(t)g_1(t)\exp\big(\int_0^t\gamma(s)ds\big) \\
% &= (-\gamma(t)g_1(t)+\beta(t))\exp\big(\int_0^t\gamma(s)ds\big) +\\ &\gamma(t)g_1(t)\exp\big(\int_0^t\gamma(s)ds\big) \\
% &= \beta(t)\exp\big(\int_0^t\gamma(s)ds\big),
% \end{aligned}
% \end{equation}
% which implies that
% $
% g_1(t) = \exp\big(-\int_{0}^{t}\gamma(s)ds\big)g_1(0)+\\ \int_{0}^{t}\exp\big(-\int_s^t \gamma(r)dr\big)\beta(s)ds.
% $
% Now consider the following equations for $g(t)-g_1(t)$:
% \begin{equation}\notag
% \begin{aligned}
% \dot{g}(t)-\dot{g_1}(t) &\le -\gamma(t)(g(t)-g_1(t));\\
% g(0) - g_1(0) &= 0.
% \end{aligned}
% \end{equation}
% It is impossible that $g(t')-g_1(t')>0$ for some $t'>0$ because $g(\bar{t})-g_1(\bar{t})\le 0$ once $g(\bar{t})-g_1(\bar{t})=0$ for some $\bar{t}>0$ and $g(0)-g_1(0)=0$. Hence $g(t)\le g_1(t)=\exp(-\int_{0}^{t}\gamma(s)ds)g(0)+\int_{0}^{t}\exp(-\int_s^t \gamma(r)dr)\beta(s)ds$.
Now we prove the two statements in the following:

\noindent (i). Suppose the conditions $\int_0^\infty\gamma(t)\mathrm{dt}=\infty$ and $\lim\limits_{t\to\infty}\frac{\beta(t)}{\gamma(t)}=0$ hold. Evidently, the term $u(t):=\exp(-\int_{0}^{t}\gamma(s)\mathrm{ds})g(0)$ goes to zero as $t$ goes to infinity. Then we focus on the other term
$$k(t):=\int_{0}^{t}\exp(-\int_s^t \gamma(r)\mathrm{dr})\beta(s)\mathrm{ds}.$$
Since for a sufficiently small $\epsilon>0$, there exists $t_0>0$ such that $\frac{\beta(t)}{\gamma(t)}<\epsilon$ for all $t>t_0$. {Define $\xi=\max\limits_{0\le t \le t_0}\frac{\beta(t)}{\gamma(t)}$. Then for all $t>t_0$, there holds
\begin{equation}\notag
\begin{aligned}
k(t)
% &< \xi\int_{0}^{t_0}\exp(-\int_s^t \gamma(r)\mathrm{dr})\gamma(s)\mathrm{ds}\\
% &\quad+\epsilon\int_{t_0}^{t}\exp(-\int_s^t \gamma(r)\mathrm{dr})\gamma(s)\mathrm{ds} \\
&< \xi\int_{0}^{t_0}\mathrm{d}(\exp(-\int_s^t \gamma(r)\mathrm{dr}))+\epsilon\int_{t_0}^{t}\mathrm{d}(\exp(-\int_s^t \gamma(r)\mathrm{dr})) \\
&= \xi\exp(-\int_{t_0}^t \gamma(r)\mathrm{dr})(1-\exp(-\int_0^{t_0}\gamma(r)\mathrm{dr}))+\epsilon(1 - \exp(-\int_0^t \gamma(r)\mathrm{dr}))\\
&<\xi\exp(-\int_{t_0}^t \gamma(r)\mathrm{dr})+\epsilon.
\end{aligned}
\end{equation}

% Consider the following two equations with respect to $l(t)$:
% \begin{equation}\notag
% \begin{aligned}
% \dot{l}(t) &= \alpha(t)(1-l(t));\\
% l(0) &= 0.
% \end{aligned}
% \end{equation}
% Note that $\exp(-\int_s^t \alpha(r)dr)\alpha(s)>0$, i.e., $\dot{l}(t)>0$ for $t>t_0$.

Since $\exp(-\int_{t_0}^t \gamma(r)\mathrm{dr})$ goes to zero as $t$ goes to infinity, one has $\lim\limits_{t\to\infty}k(t)=0$.} Then we have $\lim\limits_{t\to\infty}g(t)=0$.

\noindent (ii). Suppose the conditions $\int_0^\infty\gamma(t)\mathrm{dt}=\infty$ and\\ $\limsup\limits_{t\to\infty}\frac{\beta(t)}{\gamma(t)}<\infty$ hold. Then {there exist $B>0$ and $\hat{t}>0$ such that} $\frac{\beta(t)}{\gamma(t)}<B$ for all $t>\hat{t}$. Similarly, the limit of the term $u(t)=\exp(-\int_{0}^{t}\gamma(s)\mathrm{ds})g(0)$ is zero as $t$ goes to infinity, i.e., given $B>0$, there exists $t_u>0$ such that $u(t)<B$ for all $t>t_u$. Also we have $k(t) < B\int_{0}^{t}\exp(-\int_s^t \gamma(r)\mathrm{dr})\gamma(s)\mathrm{ds}<B$ for $t>\hat{t}$. Let $t_0:=\max\{\hat{t},t_u\}$. Hence, $g(t)<2B$ for $t>t_0$. Since $g(t)$ is continuous, we have $g(t)<\max\{B_1,2B\}$ for all $t\ge 0$ where $B_1=\max\limits_{0\le t\le t_0}g(t)$, i.e., $\{g(t)\}_{t\geq0}$ is bounded.
\end{proof}

\begin{lemma}\label{lem3}
Consider the flow (\ref{fixed}) and the underlying communication graph $\mathcal{G}_{\sigma(t)}$. Suppose there exists $M>0$ such that $\|\xb(t)\|\le M$ for all $t\ge 0$. Suppose $\mathcal{G}_\sigma(t)$ is uniformly jointly connected. Let $\xb_i(t)$ for all $i$ denote the state held by node $i$ of $\mathcal{G}_{\sigma(t)}$. Define $\Phi(t)=\max\limits_{1\le i,j\le N}\|\xb_i(t)-\xb_j(t)\|$ and a continuous function $\alpha:\R^{\ge 0}\to\R^+$. If $\int_0^{\infty}\alpha^2(t)\mathrm{dt}<\infty$, then $\int_{0}^{\infty}\alpha(t)\Phi(t)\mathrm{dt}<\infty$.
\end{lemma}

{
\begin{proof}
% By \cite{Shi2013Robust}, we know that there exists $C_1>0, C_2>0$ such that for all $k\ge 0$ and $kC_1\le t\le (k+1)C_1$,
% \begin{equation}\label{eq:SIAM1}
% \Phi(t)\le \Phi(kC_1)+C_2\int_{kC_1}^{(k+1)C_1}\alpha(t)\mathrm{dt},
% \end{equation}
% \begin{equation}\label{eq:SIAM2}
% \Phi((k+1)C_1) \le \beta \Phi(kC_1)+C_2\int_{kC_1}^{(k+1)C_1}\alpha(t)\mathrm{dt}
% \end{equation}
% with $0<\beta<1$. Clearly, (\ref{eq:SIAM1}) and the Cauchy$-$Schwarz inequality yield the upper bound of $\int_{0}^{\infty}\alpha(t)\Phi(t)\mathrm{dt}$ as
% \begin{align}
% &\quad\int_{0}^{\infty}\alpha(t)\Phi(t)\mathrm{dt} = \sum_{k=0}^{\infty}\int_{kC_1}^{(k+1)C_1}\alpha(t)\Phi(t)\mathrm{dt} \notag\\
% % &\le \sum_{k=0}^{\infty}\int_{kC_1}^{(k+1)C_1}\alpha(t)(\Phi(kC_1)\\
% % &+C_2\int_{kC_1}^{(k+1)C_1}\alpha(s)ds)dt\\
% &\le \sum_{k=0}^\infty\omega_k\Phi(kC_1)+ C_1C_2\int_0^{\infty}\alpha^2(t)\mathrm{dt} \label{eq:Ineq1}
% \end{align}
% with $\omega_k=\int_{kC_1}^{(k+1)C_1}\alpha(t)\mathrm{dt}$. Further using (\ref{eq:SIAM2}), one can easily see that the first term at the right side of (\ref{eq:Ineq1}) is bounded by some constant. Thus, it is true that $\int_{0}^{\infty}\alpha(t)\Phi(t)\mathrm{dt}$ is finite if $\int_0^{\infty}\alpha^2(t)\mathrm{dt}$ is finite.
{By \cite{Shi2013Robust}, we know that there exists $C_1>0, C_2>0$ such that for all $k\ge 0$ and $kC_1\le t\le (k+1)C_1$,
\begin{align}
\Phi(t)&\le \Phi(kC_1)+C_2\int_{kC_1}^{(k+1)C_1}\alpha(t)\mathrm{dt}\label{eq:SIAM1}\\
\Phi((k+1)C_1) &\le \beta \Phi(kC_1)+C_2\int_{kC_1}^{(k+1)C_1}\alpha(t)\mathrm{dt} \label{eq:SIAM2}
\end{align}
with $\beta\in(0,1)$. Define $\omega_k:=\int_{kC_1}^{(k+1)C_1}\alpha(t)\mathrm{dt}$ and $\alpha^\ast:=\sup\limits_{t\ge0} \alpha(t)$. Then the proof is completed by the following inequalities.
\begin{align}
\quad\int_{0}^{\infty}\alpha(t)\Phi(t)\mathrm{dt}&=\sum_{k=0}^{\infty}\int_{kC_1}^{(k+1)C_1}\alpha(t)\Phi(t)\mathrm{dt} \notag\\
&\overset{\mathrm a)}{\le} \sum_{k=0}^{\infty}\int_{kC_1}^{(k+1)C_1}\alpha(t)(\Phi(kC_1)+C_2\int_{kC_1}^{(k+1)C_1}\alpha(s)\mathrm{ds})\mathrm{dt} \notag\\
&= \sum_{k=0}^\infty\omega_k\Phi(kC_1)+ C_2\sum_{k=0}^\infty(\int_{kC_1}^{(k+1)C_1}\alpha(t)\mathrm{dt})^2\notag\\
&\overset{\mathrm b)}{\le} \sum_{k=0}^\infty\omega_k\Phi(kC_1)+ C_1C_2\int_0^{\infty}\alpha^2(t)\mathrm{dt}\notag\\
&\overset{\mathrm c)}{\le} \sum_{k=1}^\infty \omega_k(\beta^k\Phi(0)+C_2\sum_{r=1}^k\beta^{k-r}\omega_{r-1})+\omega_0\Phi(0)\notag\\
&\quad+ C_1C_2\int_0^{\infty}\alpha^2(t)\mathrm{dt},\notag
\end{align}
where ${\mathrm a)}$ is from (\ref{eq:SIAM1}), ${\mathrm b)}$ is due to Cauchy--Schwarz inequality,  and ${\mathrm c)}$ is from (\ref{eq:SIAM2}). This allows us to further conclude
\begin{align}
&\quad\int_{0}^{\infty}\alpha(t)\Phi(t)\mathrm{dt} \notag\\
&\le \alpha^\ast C_1\Phi(0)\sum_{k=1}^\infty\beta^k+\frac{C_2}{2}\sum_{k=1}^\infty\sum_{r=1}^k\beta^{k-r}(\omega_k^2+\omega_{r-1}^2)+\omega_0\Phi(0)+ C_1C_2\int_0^{\infty}\alpha^2(t)\mathrm{dt}\notag\\
&\le \frac{\alpha^\ast\beta C_1\Phi(0)}{1-\beta}+\frac{C_2}{1-\beta}\sum_{k=1}^\infty \omega_k^2+\omega_0\Phi(0)+ C_1C_2\int_0^{\infty}\alpha^2(t)\mathrm{dt}\notag\\
&= (\frac{C_2}{1-\beta}+C_1C_2)\int_0^\infty\alpha^2(t)\mathrm{dt}+(\frac{\alpha^\ast\beta C_1}{1-\beta}+\omega_0)\Phi(0), \notag
\end{align}
which completes the proof of the lemma.   }
\end{proof}

}

% The proofs of Lemma \ref{lem1} and Lemma \ref{lem2} can be found in Appendix.

% Lemma \ref{lem3} can be proven by the similar arguments in the proof of Lemma 6.2 in \cite{lou}
% and so we omit its proof here.

% The proof of Lemma \ref{lem3} can be found in Appendix.

\subsection{Proof of Theorem \ref{thm:1}}
The proof starts by establishing $\xb(t)$ is bounded, which is given as follows. Consider
\begin{equation}\notag
\begin{aligned}
Q_K(\xb,t) &:= \xb^\top\Mb(t)\xb \\
&= K\sum\limits_{\{i,j\}\in\mathcal{E}}[\mathbf{A}]_{ij}\|\xb_j-\xb_i\|^2+\alpha(t)\sum\limits_{i=1}^{N}|\hb_i^\top\xb_i|^2
\end{aligned}
\end{equation}
with $\xb\neq 0$. Clearly $Q_K(\xb,t)\ge 0$ and the equality holds only if $\xb_i=\xb_j$ for any $i,j$ and $\hb_i^\top\xb_i=0$ for all $i$. Because $\rank(\Hb)=m$ by hypothesis, there does not exist $\xb\neq 0$ such that $Q_K(\xb,t)=0$, i.e., $Q_K(\xb,t)>0$ for $\xb\neq 0$. Therefore, $\Mb(t)$ is positive-definite for all $t$. Similarly, $\mathbf{P}:=\Lb\otimes\Ib_m+\tilde{\Hb}$ is also positive-definite. Under Assumption \ref{ass:1} (ii), we know that there exists sufficiently large $t_0$ such that $\alpha(t)<K$ for all $t>t_0$. By Theorem 4.2.2 in \cite{horn2012matrix}, we know that $Q_K(\xb,t) \ge \alpha(t)\xb^\top \mathbf{P}\xb  \ge \alpha(t)\sigm(\mathbf{P})\|\xb\|^2$ for any $\xb$ and all $t>t_0$. Let $h(t)=\|\xb(t)\|^2$. Then
\begin{equation}\notag
\begin{aligned}
\quad\frac{\mathrm{d}}{\mathrm{dt}}h(t)
&= -2\xb(t)^\top(K(\Lb\otimes\Ib_m)+\alpha(t)\tilde{\Hb})\xb(t)+2\alpha(t)\xb(t)^\top\zb_H \\
&\le -2\alpha(t)\sigm(\mathbf{P})\|\xb(t)\|^2+2\alpha(t)\|\xb(t)\|\|\zb_H\|
\end{aligned}
\end{equation}
for $t>t_0$. Consider
\begin{align}
&\quad\frac{\mathrm{d}}{\mathrm{dt}}\sqrt{h(t)} = \frac{\dot{h}(t)}{2\sqrt{h(t)}}\le -\alpha(t)\sigm(\mathbf{P})\sqrt{h(t)}+\alpha(t)\|\zb_H\|,\ t\ge t_0. \label{eq:bounded_state_proved}
\end{align}
By Lemma \ref{lem2}.(ii), identifying $g(t)$ with $\sqrt{h(t)}$, we have that $\sqrt{h(t)}=\|\xb(t)\|$ is bounded for $t>t_0$. Due to the continuity of $\xb(t)$, $\|\xb(t)\|$ is bounded for all $t\ge 0$.

For the second step of the proof, we first denote $\bar{\xb}(t) := \frac{1}{N}\sum\limits_{i=1}^{N}\xb_i(t)$ and $\bar{\xb}^\diamond(t) := \mathbf{1}_N\otimes\bar{\xb}(t)$. By simple calculation, it can be shown that $\dot{\bar{\xb}}^\diamond(t)=\mathbf{1}_N\otimes(\frac{1}{N}\sum\limits_{i=1}^{N}\dot{\xb}_i(t))=-\mathbf{1}_N\otimes(\frac{\alpha(t)}{2N}\sum\limits_{i=1}^{N}\nabla f_i(\xb_i))$. Then by \cite{horn2012matrix}
\begin{align}
&\quad\frac{\mathrm{d}}{\mathrm{dt}}\|\xb(t)-\bar{\xb}^\diamond(t)\|^2\notag\\
&= 2\langle\xb(t)-\bar{\xb}^\diamond(t),\ \dot{\xb}(t)-\dot{\bar{\xb}}^\diamond(t)\rangle \notag\\
&= 2\langle\xb(t)-\bar{\xb}^\diamond(t),\ -K(\Lb\otimes\Ib_m)\xb(t) -\alpha(t)\tilde{\Hb}\xb(t)  +\alpha(t)\zb_H+\mathbf{1}_N\otimes(\frac{\alpha(t)}{2N}\sum\limits_{i=1}^{N}\nabla f_i(\xb_i(t)))\rangle \notag\\
&= 2\langle\xb(t)-\bar{\xb}^\diamond(t),-K(\Lb\otimes\Ib_m)(\xb(t)-\bar{\xb}^\diamond(t))\rangle + \beta(t) \notag\\
&\le -2\sigsecm(\Lb) K\|\xb(t)-\bar{\xb}^\diamond(t)\|^2+\beta(t),\label{eq:consensus_inequality}
\end{align}
where
\begin{equation}
\begin{aligned}\notag
&\beta(t)=2\alpha(t)\langle\xb(t)-\bar{\xb}^\diamond(t),\zb_H-\tilde{\Hb}\xb(t)+\mathbf{1}_N\otimes(\frac{1}{2N}\sum\limits_{i=1}^{N}\nabla f_i(\xb_i(t)))\rangle.
\end{aligned}
\end{equation}
Under Assumption \ref{ass:1} (ii) and by the claim that $\|\xb(t)\|$ is bounded, we know that $\lim\limits_{t\to\infty}\beta(t)=0$. By Lemma \ref{lem2}.(i), $\lim\limits_{t\to\infty}\|\xb(t)-\bar{\xb}^\diamond(t)\|^2=0$, i.e., the dynamical system (\ref{eq:x_ori}) achieves a consensus.

Now we turn to the last step of the proof and analyze the relationship between $\bar{\xb}(t)$ and the optimal point $\yb^\ast$. Let
$$\omega(t)=\frac{\alpha(t)}{N}\langle\bar{\xb}(t)-\yb^\ast,\nabla f(\bar{\xb}(t))-\sum\limits_{i=1}^{N}\nabla f_i(\xb_i(t))\rangle.$$
By Lemma \ref{lem1}, $f(\yb)$ is $2\sigm(\Hb^\top\Hb)$-strongly convex, and there holds
{
\begin{align}
\quad \frac{\mathrm{d}}{\mathrm{dt}}\|\bar{\xb}(t)-\yb^\ast\|^2 &= 2\langle\bar{\xb}(t)-\yb^\ast,\dot{\bar{\xb}}(t)\rangle  \notag \\
% &= 2\langle\bar{\xb}(t)-\yb^\ast,-\frac{\alpha(t)}{2N}\sum\limits_{i=1}^{N}\nabla f_i(\xb_i(t))\rangle  \notag \\
&= -\frac{\alpha(t)}{N}\langle\bar{\xb}(t)-\yb^\ast,\sum\limits_{i=1}^{N}\nabla f_i(\xb_i(t))\rangle \notag \\
&= -\frac{\alpha(t)}{N}\langle\bar{\xb}(t)-\yb^\ast,\nabla f(\bar{\xb}(t))\rangle+\omega(t)\notag\\
% &\quad -\frac{\alpha(t)}{N}\langle\bar{\xb}(t)-\yb^\ast,\sum\limits_{i=1}^{N}\nabla f_i(\xb_i(t))-\nabla f(\bar{\xb}(t))\rangle \notag\\
% &= 2\langle\bar{\xb}(t)-\yb^\ast,-\frac{\alpha(t)}{2N}\nabla f(\bar{\xb}(t)) +\frac{\alpha(t)}{2N}(\nabla f(\bar{\xb}(t))-\sum\limits_{i=1}^{N}\nabla f_i(\xb_i(t)))\rangle \notag \\
&\le  -\frac{\alpha(t)}{N}(f(\bar{\xb}(t))-f(\yb^\ast)+\sigm(\Hb^\top\Hb)\|\bar{\xb}(t)-\yb^\ast\|^2) + \omega(t) \label{eq:diff_xbar_yast} \\
&\le -\frac{2\sigm(\Hb^\top\Hb)\alpha(t)}{N}\|\bar{\xb}(t)-\yb^\ast\|^2 + \omega(t).  \label{eq:diff_xbar_yast1}
\end{align}}
Since $\lim\limits_{t\to\infty}(\bar{\xb}(t)-\xb_i(t))=0$, namely $\lim\limits_{t\to\infty}(\nabla f(\bar{\xb}(t))-\sum\limits_{i=1}^{N}\nabla f_i(\xb_i(t)))=0$, we have $\lim\limits_{t\to\infty}\|\bar{\xb}(t)-\yb^\ast\|^2=0$ by Lemma \ref{lem2}.(i), i.e., (\ref{eq:x_ori}) reaches a consensus and finally all nodes hold the value of the least-squares solution to (\ref{eq:linear_equation}), which completes the proof.

{

\subsection{Proof of Theorem \ref{thm:convergence_rate}}

We continue to use the definitions of $\beta(t),\bar{\xb}(t),\bar{\xb}^\diamond(t),\omega(t)$ in the proof of Theorem \ref{thm:1}.

\noindent (i) Let $\alpha(t)=\MO(\frac{1}{t})$. Due to the boundedness of $\|\xb(t)\|$ proved by (\ref{eq:bounded_state_proved})
\begin{align}
\beta(t) &= \MO(\alpha(t)\|\xb(t)-\bar{\xb}^\diamond(t)\|)\label{eq:beta_relations}\\
&=\MO\bigg(\frac{1}{t}\bigg).\label{eq:bounded_beta}
\end{align}
% Then it follows (\ref{eq:consensus_inequality}) and (\ref{eq:bounded_beta})
% \begin{equation}\label{eq:consensus_ineq1}
% \frac{d}{dt}\|\xb(t)-\bar{\xb}^\diamond(t)\|^2 \le -\sigsecm(\Lb) K\|\xb(t)-\bar{\xb}^\diamond(t)\|^2+\rho_1\alpha(t).
% \end{equation}
% Evidently
% \begin{equation}\notag
% \frac{d}{dt}\sum\limits_{i=1}^N\|\xb_i(t)-\bar{\xb}(t)\| \le \rho_2\frac{d}{dt}\|\xb(t)-\bar{\xb}^\diamond(t)\|^2
% \end{equation}
% with $\rho_2:=$
% According to the proof of Theorem \ref{thm:1}, the following inequality holds.
% \begin{equation}\label{eq:consensus_inequality}
% \frac{d}{dt}\|\xb(t)-\bar{\xb}^\diamond(t)\|^2\le -2\sigsecm(\Lb) K\|\xb(t)-\bar{\xb}^\diamond(t)\|^2+\beta(t),
% \end{equation}
By applying Lemma \ref{lem2} to (\ref{eq:consensus_inequality}) and based on (\ref{eq:bounded_beta}), one has
\begin{equation}\label{eq:consensus_ineq2}
\|\xb(t)-\bar{\xb}^\diamond(t)\|^2 = \int_0^t \MO\bigg(\frac{e^{2\sigsecm(\Lb) K(s-t)}}{s}\bigg) \mathrm{ds}.
\end{equation}
Clearly (\ref{eq:consensus_ineq2}) with Lemma \ref{lem:int_bound} yields
\begin{equation}\label{eq:consensus_ineq_results}
\|\xb(t)-\bar{\xb}^\diamond(t)\|^2 = \MO\bigg(\frac{1}{t}\bigg).
\end{equation}
It can be noticed that (\ref{eq:beta_relations}) shows $\beta(t)$ is bounded by a function of $\|\xb(t)-\bar{\xb}^\diamond(t)\|$. Hence (\ref{eq:consensus_ineq_results}) leads to a tighter bound of $\beta(t)$ than (\ref{eq:bounded_beta})
$$ \beta(t) = \MO\bigg(\frac{1}{t^{\frac{3}{2}}}\bigg).$$
Based on (\ref{eq:beta_relations}), by recursively applying Lemma \ref{lem:int_bound} and Lemma \ref{lem2} on (\ref{eq:consensus_inequality}) with constantly updated upper bounds of $\beta(t)$ initialized by (\ref{eq:bounded_beta}), we can obtain a sequence of bounds on $\|\xb(t)-\bar{\xb}^\diamond(t)\|^2$ as following.
\begin{align}
% \{\beta(t) &\in \MO (t^{a_r})\}_{r=1,2,\dots}\notag\\
\|\xb(t)-\bar{\xb}^\diamond(t)\|^2 = \MO (t^{a_r}),\ r=1,2,\dots, \label{eq:bounds_seq1}
\end{align}
where
\begin{equation}\notag
a_{r+1} = \frac{1}{2} a_r - 1,\ a_1 = -1.
\end{equation}
Clearly, $a_r$ in (\ref{eq:bounds_seq1}) goes to $-2$ as $r$ go to infinity. Then there holds
\begin{equation}\label{eq:consensus_ineq_results2}
\|\xb(t)-\bar{\xb}^\diamond(t)\|^2 = \MO\bigg(\frac{1}{t^2}\bigg).
\end{equation}
From the Cauchy--Schwarz inequality and (\ref{eq:consensus_ineq_results2})
{\begin{align}
\omega(t) &= \frac{2\alpha(t)}{N}(\bar{\xb}(t)-\yb^\ast)^\top\sum\limits_{i=1}^N \hb_i\hb_i^\top(\bar{\xb}(t)-\xb_i(t)) \notag\\
&\le \frac{2\alpha(t)}{N}\|\bar{\xb}(t)-\yb^\ast\|\sum\limits_{i=1}^N \|\hb_i\|^2\|\bar{\xb}(t)-\xb_i(t)\|\notag\\
&\le \rho\alpha(t)\|\bar{\xb}(t)-\yb^\ast\|\|\xb(t)-\bar{\xb}^\diamond(t)\| \label{eq:omega_relations}\\
&=\MO\big(t^{-2}\|\bar{\xb}(t)-\yb^\ast\|\big), \label{eq:omega_eq1}
\end{align}}
where
$$\rho:=\max\{2N^{-\frac{1}{2}}\|\hb_i\|^2:i\in\mathcal{V}\}.$$
% Recall in the proof of Theorem \ref{thm:1} we have
% \begin{equation}\label{eq:diff_xbar_yast1_thm2}
% \frac{d}{dt}\|\bar{\xb}(t)-\yb^\ast\|^2  \le -\frac{2\sigm(\Hb^\top\Hb)\alpha(t)}{N}\|\bar{\xb}(t)-\yb^\ast\|^2 + \omega(t),
% \end{equation}
{We apply Lemma \ref{lem2} on (\ref{eq:diff_xbar_yast1}) using the bound in (\ref{eq:omega_eq1}) and obtain
\begin{align}
\|\bar{\xb}(t)&-\yb^\ast\|^2=\MO\bigg(t^{-\frac{2\sigm(\Hb^\top\Hb)}{N}}\bigg)+\MO\bigg(t^{-\frac{2\sigm(\Hb^\top\Hb)}{N}}\bigg)\cdot\int_0^t \MO\bigg(s^{\frac{2\sigm(\Hb^\top\Hb)}{N}-2}\cdot\|\bar{\xb}(s)-\yb^\ast\|\bigg)\mathrm{ds}.\label{eq:integral_cases}
\end{align}
Depending on whether
$$s^{\frac{2\sigm(\Hb^\top\Hb)}{N}-2}\cdot\|\bar{\xb}(s)-\yb^\ast\|=\MO(s^{-1}),$$
the integral part in (\ref{eq:integral_cases}) falls into two different function classes. Therefore, we will discuss the bound of $\|\bar{\xb}(t)-\yb^\ast\|^2$ in two cases.

\noindent (a) We assume $\sigm(\Hb^\top\Hb)\neq N$. Define a set $\mU\subset[1,2)$ with
$$ \mU := \bigg\{\sum\limits_{i=1}^{r}\bigg(\frac{1}{2}\bigg)^{i-1}:r=2,3,\dots\bigg\}\mcup\{1\}.$$
We will see the proof of (a) can be achieved under two complementary scenarios.

\noindent [Scenario 1] Suppose $\frac{2\sigm(\Hb^\top\Hb)}{N}\in\R^+\setminus(\mU\mcup\{2\})$. From (\ref{eq:integral_cases}) with the fact $\|\bar{\xb}(t)-\yb^\ast\|=\MO(1)$
\begin{equation}
\|\bar{\xb}(t)-\yb^\ast\|^2=  \MO\bigg(\frac{1}{t^{\frac{2\sigm(\Hb^\top\Hb)}{N}}}+\frac{1}{t}\bigg).\label{eq:1_t_convergence_rate1}
\end{equation}
Define two sequences $\{b_r\}_{r=1,2,\dots}$ and $\{\hat{b}_r\}_{r=1,2,\dots}$ with
\begin{align}
b_{r+1} &= \frac{1}{2}b_r-1,\ b_1=-\frac{2\sigm(\Hb^\top\Hb)}{N}\notag\\
\hat{b}_{r+1} &= \frac{1}{2}\hat{b}_r-1,\ \hat{b}_1=-1.\notag
\end{align}
Direct verification shows
\begin{align}
b_r &\neq -\frac{2\sigm(\Hb^\top\Hb)}{N},\ \forall r\ge 2\label{eq:gua1}\\
\hat{b}_r &\neq -\frac{2\sigm(\Hb^\top\Hb)}{N},\ \forall r\ge 1. \label{eq:gua2}
\end{align}
It is evident (\ref{eq:gua1}) and (\ref{eq:gua2}) guarantee that no integral of $\MO(s^{-1})$ arises the following iteration process. Clearly
\begin{align}
\|\bar{\xb}(t)-\yb^\ast\|^2&\overset{\rm a)}{=}\MO\bigg(t^{-\frac{2\sigm(\Hb^\top\Hb)}{N}}\bigg)+\MO\bigg(t^{-\frac{2\sigm(\Hb^\top\Hb)}{N}}\bigg)\cdot\int_0^t \MO\bigg(s^{\frac{2\sigm(\Hb^\top\Hb)}{N}-2}\cdot\big(s^{-\frac{\sigm(\Hb^\top\Hb)}{N}}+s^{-\frac{1}{2}}\big) \bigg)\mathrm{ds}\notag\\
&\overset{\rm b)}{=}\MO\bigg(t^{-\frac{2\sigm(\Hb^\top\Hb)}{N}}+t^{-\frac{\sigm(\Hb^\top\Hb)}{N}-1}+t^{-\frac{3}{2}}\bigg),\label{eq:demo1}
\end{align}
where a) comes from (\ref{eq:integral_cases}) and (\ref{eq:1_t_convergence_rate1}), and b) is obtained by direct calculation. We apply a series of the recursions as from (\ref{eq:1_t_convergence_rate1}) to (\ref{eq:demo1}) and obtain the following bound.
% Analogously to (\ref{eq:bounded_beta}) and (\ref{eq:consensus_ineq_results}), (\ref{eq:omega_eq1}) and (\ref{eq:1_t_convergence_rate1}) initialize an iteratively updated bound sequence for $\|\bar{\xb}(t)-\yb^\ast\|^2$ on the basis of (\ref{eq:integral_cases}), leading to
\begin{align}
\|\bar{\xb}(t)-\yb^\ast\|^2
% &= \MO\bigg(\frac{1}{t^{\frac{2\sigm(\Hb^\top\Hb)}{N}}}+\frac{1}{t^{\frac{\sigm(\Hb^\top\Hb)}{N}{+1}}}+\frac{1}{t^{\frac{3}{2}}}\bigg)\notag\\
% &= \cdots\notag\\
&=\MO\bigg(\sum\limits_{r=1}^\infty t^{b_r} + t^{\hat{b}_\infty}\bigg)=\MO\bigg(\frac{1}{t^{\min(\frac{2\sigm(\Hb^\top\Hb)}{N},2)}}\bigg),\label{eq:results_caseA}
\end{align}
where	
\begin{equation}\notag
\begin{aligned}
\hat{b}_\infty := {\lim\limits_{r\to\infty}\hat{b}_r}.
\end{aligned}
\end{equation}
\noindent [Scenario 2] Suppose $\frac{2\sigm(\Hb^\top\Hb)}{N}\in \mU$. Then there exists $r^\ast\in\{1,2,\dots\}$ such that
$$\hat{b}_{r^\ast}=-\frac{2\sigm(\Hb^\top\Hb)}{N}.$$
For ease of presentation, we define $\hat{b}_0 = 0$. Similarly to the process of obtaining (\ref{eq:results_caseA}), we apply $r^\ast$ rounds of iterations based on (\ref{eq:integral_cases}), and arrive at
\begin{align}
\|\bar{\xb}(t)-\yb^\ast\|^2
% &= \MO\big(t^{b_1}\big) + \MO\bigg(t^{-\frac{2\sigm(\Hb^\top\Hb)}{N}}\bigg)\notag\\
% &\quad \cdot\int_0^t \MO\bigg(s^{\frac{2\sigm(\Hb^\top\Hb)}{N}-2}\cdot s^{\hat{b}_0}\bigg)\mathrm{ds}\notag\\
% &= \cdots \notag\\
&= \MO\bigg(\sum\limits_{r=1}^{r^\ast} t^{b_r}\bigg) + \MO\bigg(t^{-\frac{2\sigm(\Hb^\top\Hb)}{N}}\bigg) \cdot\int_0^t \MO\bigg(s^{\frac{2\sigm(\Hb^\top\Hb)}{N}-2}\cdot s^{\hat{b}_{r^\ast-1}}\bigg)\mathrm{ds}\notag\\
&= \MO\bigg(\sum\limits_{r=1}^{r^\ast} t^{b_r}+t^{-\frac{2\sigm(\Hb^\top\Hb)}{N}}\log t\bigg).\label{eq:r_star_stop1}
% \notag\\
% &= \cdots\notag\\
% &=\MO\bigg(\sum\limits_{r=1}^\infty t^{b_r} + t^{\hat{b}_\infty}\bigg)\notag\\
% &=\MO\bigg(\frac{1}{t^{\min(\frac{2\sigm(\Hb^\top\Hb)}{N},2)}}\bigg)\notag
\end{align}
Noticing the fact that the scenario hypothesis $\frac{2\sigm(\Hb^\top\Hb)}{N}\in[1,2)$, we claim there exists
$$\delta\in\bigg(0,2-\frac{2\sigm(\Hb^\top\Hb)}{N}\bigg)$$
such that
\begin{equation}\label{eq:bigo_relations}
\log t = \MO(t^\delta).
\end{equation}
Then it follows (\ref{eq:r_star_stop1}) and (\ref{eq:bigo_relations})
\begin{equation}\label{eq:r_star_continue1}
\|\bar{\xb}(t)-\yb^\ast\|^2=\MO\bigg(\sum\limits_{r=1}^{r^\ast} t^{b_r}+t^{\delta-\frac{2\sigm(\Hb^\top\Hb)}{N}}\bigg).
\end{equation}
Define a sequence $\{d_r\}_{r=1,2,\dots}$ with
$$ d_{r+1} = \frac{1}{2}d_r-1,\ d_1=\delta-\frac{2\sigm(\Hb^\top\Hb)}{N}.$$
Then it can be easily verified
\begin{equation}\label{eq:guarantee_no_log1}
d_2<-\frac{2\sigm(\Hb^\top\Hb)}{N}<d_1,
\end{equation}
which implies that there is no element in $\{d_r\}_{r=1,2,\dots}$ equal to $-\frac{2\sigm(\Hb^\top\Hb)}{N}$. Now we continue the iteration from (\ref{eq:r_star_continue1}), during which process (\ref{eq:guarantee_no_log1}) guarantees no integral of $\MO(s^{-1})$ arises. Infinite iterations indicate that the following bound holds.
\begin{align}
\|\bar{\xb}(t)-\yb^\ast\|^2 &= \MO\bigg(\sum\limits_{r=1}^{\infty} t^{b_r}+t^{d_\infty}\bigg)=\MO\bigg(\frac{1}{t^{\frac{2\sigm(\Hb^\top\Hb)}{N}}}\bigg)\label{eq:results_caseC}
\end{align}
with
$$ d_\infty := \lim\limits_{r\to\infty} d_r.$$
Evidently, the proof of (a) is completed by (\ref{eq:results_caseA}) and (\ref{eq:results_caseC}).

{\noindent (b) We assume $\sigm(\Hb^\top\Hb)=N$. Similarly, (\ref{eq:integral_cases}) gives
\begin{equation}\label{eq:1_t_convergence_rate1_caseB}
\|\bar{\xb}(t)-\yb^\ast\|^2 =\MO \bigg(\frac{1}{t^2}+\frac{1}{t}\bigg).
\end{equation}
Starting from (\ref{eq:1_t_convergence_rate1_caseB}) and based on (\ref{eq:integral_cases}), we obtain
\begin{align}
\|\bar{\xb}(t)-\yb^\ast\|^2\notag
&= \MO\big(t^{-2}\big)+\MO\big(t^{-2}\big)\int_0^t\MO\big(s^{-1}+s^{-\frac{1}{2}}\big)\mathrm{ds}          \notag\\
&= \MO\big(t^{-2}\big)+\MO\big(t^{-2}\log t\big)+\MO\big(t^{-\frac{3}{2}}\big).  \label{eq:demo2}
\end{align}
Again, we repeat the process from (\ref{eq:1_t_convergence_rate1_caseB}) to (\ref{eq:demo2}) recursively and obtain
\begin{align}
\|\bar{\xb}(t)-\yb^\ast\|^2
&=\MO\bigg(t^{-2}+t^{-2}\sum\limits_{r=1}^\infty (\log t)^{c_r} + t^{\hat{b}_\infty}\bigg)=\MO\bigg(\frac{(\log t)^2}{t^2}\bigg),\label{eq:results_caseB}
\end{align}
where
$$ c_{r+1} = \frac{1}{2} c_r + 1,\ c_1 = 1. $$}
Clearly, (\ref{eq:results_caseB}) completes the proof of (b).}
% With (\ref{eq:r_star_stop1}) and (\ref{eq:omega_relations}), $\omega(t)$ now has an updated bound
% \begin{equation}\label{eq:omega_bound_caseC1}
% \omega(t) = \MO \bigg(t^{-2}\sum\limits_{r=1}^{r^\ast} t^{\frac{b_r}{2}}+t^{-\frac{\sigm(\Hb^\top\Hb)}{N}-\frac{3}{2}}\bigg)
% \end{equation}
% Starting from (\ref{eq:omega_bound_caseC1}), Lemma \ref{lem2} can continue to be recursively applied to (\ref{eq:diff_xbar_yast1}) and finally gives
% \begin{align}
% \|\bar{\xb}(t)-\yb^\ast\|^2&=\MO\bigg(t^{-\frac{2\sigm(\Hb^\top\Hb)}{N}}\bigg)+\MO\bigg(t^{-\frac{2\sigm(\Hb^\top\Hb)}{N}}\bigg)\notag\\
% &\quad\cdot\int_0^t \MO\bigg(s^{\frac{2\sigm(\Hb^\top\Hb)}{N}-2}\bigg)\mathrm{ds}
% \end{align}}

% (\ref{eq:proof_of_5}) completes the proof of (\ref{eq:convergence_rate}).

\noindent (ii) Let $\alpha(t)=\MO(\frac{1}{t^{\lambda}})$. Immediately there holds
\begin{equation}\label{eq:beta_bound_theta}
\beta(t)=\MO\bigg(\frac{1}{t^\lambda}\bigg).
\end{equation}
Starting from (\ref{eq:beta_bound_theta}), similar recursive applications of Lemma \ref{lem:int_bound} and Lemma \ref{lem2} on (\ref{eq:consensus_inequality}) result in
\begin{align}
\|\xb(t)-\bar{\xb}^\diamond(t)\|^2 &= \int_0^t\MO \bigg(\frac{e^{2\sigsecm(\Lb) K(s-t)}}{s^\lambda}\bigg) \mathrm{ds} = \MO\bigg(\frac{1}{t^\lambda}\bigg)\notag\\
\|\xb(t)-\bar{\xb}^\diamond(t)\|^2 &= \int_0^t\MO \bigg(\frac{e^{2\sigsecm(\Lb) K(s-t)}}{s^{\frac{3}{2}\lambda}}\bigg) \mathrm{ds} = \MO\bigg(\frac{1}{t^{\frac{3}{2}\lambda}}\bigg)\notag\\
&\cdots\notag\\
\|\xb(t)-\bar{\xb}^\diamond(t)\|^2 &= \MO\bigg(\frac{1}{t^{2\lambda}}\bigg).\label{eq:consensus_ineq_results2_theta}
\end{align}
It follows (\ref{eq:consensus_ineq_results2_theta}) and the fact $\|\bar{\xb}(t)-\yb^\ast\|=\MO(1)$
\begin{equation}\label{eq:omega_bound_theta}
\omega(t) = \MO\big(\alpha(t)\|\xb(t)-\bar{\xb}^\diamond(t)\|\|\bar{\xb}(t)-\yb^\ast\|\big) = \MO\bigg(\frac{1}{t^{2\lambda}}\bigg).
\end{equation}
With (\ref{eq:omega_bound_theta}) inserted in (\ref{eq:diff_xbar_yast1}), Lemma \ref{lem2} and simple change of variables {yield}
\begin{align}
\|\bar{\xb}(t)-\yb^\ast\|^2 &= \int_0^t \MO\bigg(\frac{e^{\frac{2\sigm(\Hb^\top\Hb)}{N(1-\lambda)}(s^{1-\lambda}-t^{1-\lambda})}}{s^{2\lambda}}\bigg)\mathrm{ds}\notag\\
&= \int_0^{t^{1-\lambda}} \MO\bigg(\frac{e^{\frac{2\sigm(\Hb^\top\Hb)}{N(1-\lambda)}(s-t^{1-\lambda})}}{s^{\frac{\lambda}{1-\lambda}}}\bigg)\mathrm{ds}.\label{eq:integral_transformed}
\end{align}
Clearly, one obtains by applying Lemma \ref{lem:int_bound} on (\ref{eq:integral_transformed})
\begin{equation}\label{eq:theta_case_initialize}
\|\bar{\xb}(t)-\yb^\ast\|^2 = \MO\bigg(\frac{1}{t^\lambda}\bigg).
\end{equation}
Again starting from (\ref{eq:theta_case_initialize}), recursive applications of Lemma \ref{lem:int_bound} and Lemma \ref{lem2} on (\ref{eq:diff_xbar_yast1}) gives
\begin{align}
\|\bar{\xb}(t)-\yb^\ast\|^2 &= \MO\bigg(\frac{e^{\frac{2\sigm(\Hb^\top\Hb)}{N(1-\lambda)}(s-t^{1-\lambda})}}{s^{\frac{\frac{3}{2}\lambda}{1-\lambda}}}\bigg)\mathrm{ds}=\MO\bigg(\frac{1}{t^{\frac{3}{2}\lambda}}\bigg)\notag\\
\|\bar{\xb}(t)-\yb^\ast\|^2 &= \MO\bigg(\frac{e^{\frac{2\sigm(\Hb^\top\Hb)}{N(1-\lambda)}(s-t^{1-\lambda})}}{s^{\frac{\frac{7}{4}\lambda}{1-\lambda}}}\bigg)\mathrm{ds}=\MO\bigg(\frac{1}{t^{\frac{7}{4}\lambda}}\bigg)\notag\\
&\cdots\notag\\
\|\bar{\xb}(t)-\yb^\ast\|^2 &= \MO\bigg(\frac{1}{t^{2\lambda}}\bigg),\notag\\
\end{align}
which completes the proof of (b).
}

% With (\ref{eq:1_t_convergence_rate1}), $\omega(t)$ satisfies an updated inequality
% \begin{equation}\notag
% \omega(t) \in \MO(\frac{1}{\sqrt{t^{\min(\frac{1}{2},\frac{\sigma}{N})+3}}}),
% \end{equation}
% which leads to an updated inequality of $\|\bar{\xb}(t)-\yb^\ast\|^2$ by Lemma \ref{lem2}
% \begin{equation}\notag
% \|\bar{\xb}(t)-\yb^\ast\|^2 \in \MO(\frac{1}{t^{\min(\frac{1}{2}\min(\frac{1}{2},\frac{\sigma}{N})+\frac{1}{2},\frac{\sigma}{N})}}).
% \end{equation}
% By repetitively updating these inequalities, we find
% \begin{equation}\notag
% \|\bar{\xb}(t)-\yb^\ast\|^2 \in \MO(\frac{1}{t^{\min(\frac{1}{2}\min(\dots(\frac{1}{2}\min(\frac{1}{2},\frac{\sigma}{N})+\frac{1}{2},\frac{\sigma}{N})\dots)+\frac{1}{2},\frac{\sigma}{N})}}),
% \end{equation}
% which completes the proof of (\ref{eq:convergence_rate}).

\subsection{Proof of Theorem \ref{thm:2}}

Denote the averaged state at time $t$ by $\bar{\xb}(t) = \frac{1}{N}\sum\limits_{i=1}^{N}\xb_i(t)$ and $\bar{\xb}^\diamond(t) = \mathbf{1}_N\otimes\bar{\xb}(t)$. Denote $h(t)=\|\xb(t)\|^2$. Let $\Lb_{\sigma(t)}$ be the Laplacian of the graph $\mathcal{G}_{\sigma(t)}\in\mathcal{Q}^\ast$. Let $\mathbf{P}_{\sigma(t)}=\Lb_{\sigma(t)}\otimes\Ib_m+\tilde{\Hb}$. By a minor variant of a step in the proof of Theorem \ref{thm:1}, one has
\begin{equation}\notag
\begin{aligned}
\frac{\mathrm{d}}{\mathrm{dt}}\sqrt{h(t)} = \frac{\dot{h}(t)}{2\sqrt{h(t)}}\le -\alpha(t)\sigm(\mathbf{P}_{\sigma(t)})\sqrt{h(t)}+\alpha(t)\|\zb_H\|,\ t\ge t_0.
\end{aligned}
\end{equation}
Since $|\mathcal{Q}^\ast|<\infty$, the quantity $\min\limits_{t\ge 0}\sigm(\mathbf{P}_{\sigma(t)})=\sigm^\ast$ is well-defined and positive. Then it follows
\begin{equation}\notag
\begin{aligned}
\frac{\mathrm{d}}{\mathrm{dt}}\sqrt{h(t)} = \frac{\dot{h}(t)}{2\sqrt{h(t)}}\le -\alpha(t)\sigm^\ast\sqrt{h(t)}+\alpha(t)\|\zb_H\|,\ t\ge t_0.
\end{aligned}
\end{equation}
Thus a conclusion can be drawn that $\|\xb(t)\|$ is bounded. Similarly
\[
\frac{\mathrm{d}}{\mathrm{dt}}\|\xb(t)-\bar{\xb}^\diamond(t)\|^2
\le -2\sigsecm(\Lb_{\sigma(t)}) K\|\xb(t)-\bar{\xb}^\diamond(t)\|^2+\beta(t),
\]
where $\beta(t)=2\alpha(t)\langle\xb(t)-\bar{\xb}^\diamond(t),\zb_H-\tilde{\Hb}\xb(t)+\mathbf{1}_N\otimes(\frac{1}{2N}\sum\limits_{i=1}^{N}\nabla f_i(\xb_i))\rangle$. Then we select $\sigsecm^\ast=\min\limits_{t\ge 0}\sigsecm(\Lb_{\sigma(t)})$ so that
\[
\frac{\mathrm{d}}{\mathrm{dt}}\|\xb(t)-\bar{\xb}^\diamond(t)\|^2
\le -2\sigsecm^\ast K\|\xb(t)-\bar{\xb}^\diamond(t)\|^2+\beta(t).
\]
Similarly, by Lemma \ref{lem2} and the fact that $\lim\limits_{t\to\infty}\beta(t)=0$, we can conclude $$\lim\limits_{t\to\infty}\|\xb(t)-\bar{\xb}^\diamond(t)\|^2=0,$$
i.e., the system (\ref{eq:x_ori}) achieves a consensus over switching networks.

Next we prove that the consensus value is exactly the least-squares solution of (\ref{eq:linear_equation}). Let $\yb^\ast\in\mathcal{Y}_{\rm LS}$. Recall in (\ref{eq:diff_xbar_yast}) we have
\begin{equation}\label{eq:diff_xbar_yast2}
\frac{\mathrm{d}}{\mathrm{dt}}\|\bar{\xb}(t)-\yb^\ast\|^2 \le  -\frac{\alpha(t)}{N}(f(\bar{\xb}(t))-f(\yb^\ast)) + \omega(t),
\end{equation}
where
\begin{equation}\notag
\begin{aligned}
\omega(t) &= \frac{\alpha(t)}{N}\langle\bar{\xb}(t)-\yb^\ast,\nabla f(\bar{\xb}(t))-\sum\limits_{i=1}^{N}\nabla f_i(\xb_i(t))\rangle\\
&= \frac{\alpha(t)}{N}\langle\bar{\xb}(t)-\yb^\ast, \sum_{i=1}^N\hb_i\hb_i^\top(\xb_i(t)-\bar{\xb}(t))\rangle.
\end{aligned}
\end{equation}
By simple calculation and the fact that $\|\xb(t)\|$ is bounded, it can be obtained that
\begin{equation}\notag
\begin{aligned}
\left|\omega(t)\right| &\le \frac{\alpha(t)}{N}\|\bar{\xb}(t)-\yb^\ast\|\sum_{i=1}^N \|\hb_i\hb_i^\top\|\|\xb_i(t)-\bar{\xb}(t)\|\\
&\le \frac{\alpha(t)\Phi(t)}{N}\|\bar{\xb}(t)-\yb^\ast\|\sum_{i=1}^N \|\hb_i\hb_i^\top\|\\
&=\MO(\alpha(t)\Phi(t)),
\end{aligned}
\end{equation}
where $\Phi(t)=\max\limits_{1\le i,j\le N}\|\xb_i(t)-\xb_j(t)\|$. By Lemma \ref{lem3}
$$\int_0^\infty\left|\omega(t)\right|\mathrm{dt}<\infty,$$
which implies
$$\int_0^\infty\omega(t)\mathrm{dt}<\infty.$$
Note that the constantly connected graph
considered in this theorem is clearly uniformly jointly connected. Based on (\ref{eq:diff_xbar_yast2}), we have
\begin{equation}\label{eq:9}
\begin{aligned}
\frac{1}{N}\int_0^t\alpha(s)(f(\bar{\xb}(s))-f(\yb^\ast))\mathrm{ds}\le \|\bar{\xb}(0)-\yb^\ast\|^2-\|\bar{\xb}(t)-\yb^\ast\|^2 + \int_0^t\omega(s)\mathrm{ds}.
\end{aligned}
\end{equation}
Since $\xb(t)$ is bounded and $\int_0^\infty\omega(t)\mathrm{dt}<\infty$, the right-hand side of (\ref{eq:9}) is less than infinity, which implies
$$\int_0^\infty\alpha(s)(f(\bar{\xb}(s))-f(\yb^\ast))\mathrm{ds}<\infty.$$
Since $\int_0^\infty\alpha(s)\mathrm{ds}=\infty$, $\liminf\limits_{s\to\infty}(f(\bar{\xb}(s))-f(\yb^\ast))=0$. Since the states $\xb_i(t)$ for all $i$ are bounded, we can find a sequence $\{s_k\}_{k\ge 0}$ such that
$$\lim\limits_{k\to\infty}f(\bar{\xb}(s_k))=f(\yb^\ast).$$
By Bolzano-Weierstrass theorem, we select $\{s_{k_r}\}_{r\ge 0}$ as a subsequence of $\{s_k\}_{k\ge 0}$ such that $\lim\limits_{r\to\infty}\bar{\xb}(s_{k_r})=\hat{\yb}$ for some $\hat{\yb}$. It is obvious that $f(\hat{\yb})=f(\yb^\ast)$, i.e. $\hat{\yb}\in\mathcal{Y}$ is also an optimal solution. Moreover, by replacing $\yb^\ast$ with $\hat{\yb}$ in (\ref{eq:diff_xbar_yast2}), we have {by the convexity of the function $f$}
\begin{equation}\label{eq:10}
\frac{\mathrm{d}}{\mathrm{dt}}\|\bar{\xb}(t)-\hat{\yb}\|^2 \le \omega(t) \le \left|\omega(t)\right|.
\end{equation}
In order to prove by contradiction that $\|\bar{\xb}(t)-\hat{\yb}\|^2$ is convergent, we suppose, by the boundedness of $\bar{\xb}(t)$, that there exist sequences $\{t_{s_k}\},\{t_{r_k}\}$ satisfying that
\begin{equation}\notag
\begin{aligned}
l_1 &:= \lim\limits_{k\to\infty}\|\bar{\xb}(t_{s_k})-\hat{\yb}\|^2\\
l_2 &:= \lim\limits_{k\to\infty}\|\bar{\xb}(t_{r_k})-\hat{\yb}\|^2,
\end{aligned}
\end{equation}
respectively and $l_1\neq l_2$. We also assume, without loss of generality, $l_1-l_2=\epsilon_0>0$. Then by (\ref{eq:10}) we have
$$l_1-l_2=\lim\limits_{k\to\infty}\int_{t_{r_k}}^{t_{s_k}}\frac{\mathrm{d}}{\mathrm{dt}}\|\bar{\xb}(t)-\hat{\yb}\|^2\mathrm{dt}\le \lim\limits_{k\to\infty}\int_{t_{r_k}}^{t_{s_k}}\left|\omega(t)\right|\mathrm{dt}.$$
Since $\int_0^\infty\left|\omega(t)\right|\mathrm{dt}<\infty$ as proved above, it can be concluded that
$$\lim\limits_{k\to\infty}\int_{t_{r_k}}^{t_{s_k}}\left|\omega(t)\right|\mathrm{dt}=0,$$
i.e., there exists $k_0>0$ such that $\int_{t_{r_k}}^{t_{s_k}}\left|\omega(t)\right|\mathrm{dt}<\epsilon_0$ for all $k>k_0$. This implies $l_1-l_2<\epsilon_0$, which is contradictory to the assumption that $l_1-l_2=\epsilon_0$. Hence $\|\bar{\xb}(t)-\hat{\yb}\|^2$ is convergent. Since it has been shown that there exists a sequence $\{s_r\}_{r\ge 0}$ such that $\lim\limits_{r\to\infty}\bar{\xb}(s_r)=\hat{\yb}$, we have $\lim\limits_{t\to\infty}\|\bar{\xb}(t)-\hat{\yb}\|^2=0$. Due to the fact that the network achieves a consensus, there holds
$$\lim\limits_{t\to\infty}\xb_i(t)=\hat{\yb}$$
for all $i\in\mathcal{V}$.

\subsection{Proof of Theorem \ref{thm:3}}
By the boundedness of states and the Proposition 4.10 in \cite{Shi2013Robust},
we know that the network achieves a consensus.
Based on the hypothesis of boundedness of states and the consensus result,
we can show this theorem by the similar arguments in Theorem \ref{thm:2}.

\section{Numerical Examples}\label{sec:sims}
In this section, several numerical examples are provided to validate the results of Theorem \ref{thm:1}, \ref{thm:2}.
\subsection{Fixed Graphs}
{{{
\noindent{\bf Example 1.} Consider a $4$-node path graph $\mathcal{G}_{\rm ring}$, over which we study two linear algebraic equations with respect to $\yb\in\R^2$:
\begin{align}
\textnormal{(LE. 1) }\begin{bmatrix}
1 & 1\\
1 & 2.3\\
-0.5 & 0.8\\
0.8 & 0.2
\end{bmatrix}
\yb
&=
\begin{bmatrix}
1\\ 3\\ 2\\ -1
\end{bmatrix},\notag\\
\textnormal{(LE. 2) }\ \begin{bmatrix}
2 & 7\\
6 & 5\\
-11 & 1\\
1 & 0
\end{bmatrix}\
\yb\
&=
\begin{bmatrix}
1\\ 3\\ 2\\ -1
\end{bmatrix}.\notag
\end{align}
Both (LE. 1) and (LE. 2) yield unique least-squares solutions $\yb^\ast_1=[-1.218\ 1.869]^\top,\ \yb^\ast_2=[-0.092\ 0.361]^\top$, respectively. The resulting $\frac{2\sigm(\Hb^\top\Hb)}{N}$ values for (LE. 1) and (LE. 2) are
$$\bigg(\frac{\sigm(\Hb^\top\Hb)}{N}\bigg)_1=0.313,\ \bigg(\frac{\sigm(\Hb^\top\Hb)}{N}\bigg)_2=15.975,$$
respectively. We also introduce another equation (LE. 3) by multiplying the left-hand side of (LE. 1) with $1.7872$ so that
$$\bigg(\frac{\sigm(\Hb^\top\Hb)}{N}\bigg)_3=1.$$
With $K=100$ and some randomly chosen initial conditions $\xb(0)$,
% $\xb(0)=[2\ 1\ -1\ 3\ 5\ 2.5\ -2\ 3.9]^\top$
we run the algorithm (\ref{eq:x_ori}) with $\alpha(t)=\frac{1}{t+1}$ and then plot the trajectories of
\begin{equation}\notag
\begin{aligned}
e_1(t)&:=\big\|\sum_{i=1}^4 \xb_i(t)/4-\yb^\ast_1\big\|\\
e_2(t)&:=\big\|\sum_{i=1}^4 \xb_i(t)/4-\yb^\ast_2\big\|\\
e_3(t)&:=\big\|\sum_{i=1}^4 \xb_i(t)/4-\frac{\yb^\ast_1}{1.7872}\big\|\cdot\big(\log (t+1)\big)^{-1}
\end{aligned}
\end{equation}
in logarithmic scales in Figure \ref{fig:eg1}. As can be seen, each $\xb_i(t)$ converges to $\yb^\ast$, which is consistent with the claim of Theorem \ref{thm:1}. Further, according to the trajectories in Figure \ref{fig:eg1}, we directly calculate the slopes
$$\kappa_1=-0.313,\ \kappa_2=-0.997,\ \kappa_3=-1.040$$
for (LE. 1), (LE. 2) and (LE. 3), which implies
$$e_1(t)=\MO(\frac{1}{t^{0.313}}),\ e_2(t)=\MO(\frac{1}{t^{0.997}}),\ e_3(t)=\MO(\frac{1}{t^{1.040}}).$$
This validates the statement of Theorem \ref{thm:convergence_rate} when $\alpha(t)=\MO(\frac{1}{t})$, where the bounds of $e_1(t)$ and $e_2(t)$ are as predicted as Theorem \ref{thm:convergence_rate} (i)(a), and that of $e_3(t)$ is consistent with Theorem \ref{thm:convergence_rate} (i)(b).
}

% \begin{figure}
% \centering
% \includegraphics[width=1.5in]{Fixed1.pdf}
% \caption{A constant, connected and undirected graph $\mathcal{G}_0$ considered in Example 1 and 2.}
% \label{fig:fixed1}
% \end{figure}

\begin{figure}
\centering
\includegraphics[width=3.3in]{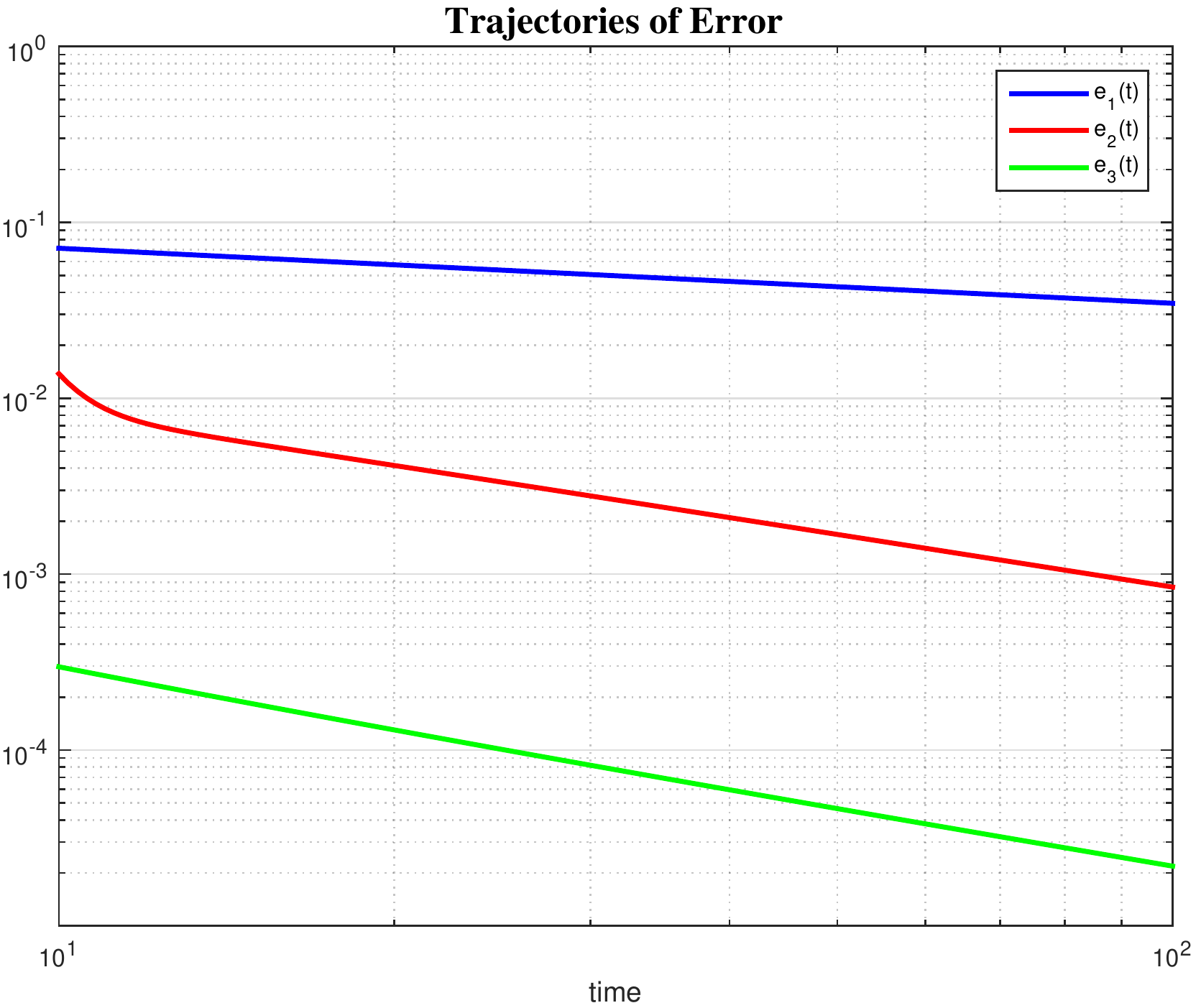}
\caption{The trajectories of $e_j(t):=\big\|\sum\limits_{i=1}^4 \xb_i(t)/4-\yb^\ast_j\big\|^2,\ j=1,2$ and $e_3(t):=\big\|\sum_{i=1}^4 \xb_i(t)/4-\frac{\yb^\ast_1}{1.7872}\big\|\cdot\big(\log (t+1)\big)^{-1}$ with $\yb^\ast_1=[-1.218\ 1.869]^\top$ and $\yb^\ast_2=[-0.092\ 0.361]^\top$ for $\alpha(t)=\frac{1}{t+1}$. The slopes are $\kappa_1=-0.313, \kappa_2=-0.997, \kappa_3=-1.040$.}
\label{fig:eg1}
\end{figure}

\medskip

\noindent{\bf Example 2.} Consider the linear equation (LE. 1) with the same $\xb(0)$ and $K$ as in Example 1. We run the algorithm (\ref{eq:x_ori}) on $\mathcal{G}_{\rm ring}$ for $\alpha(t)=\frac{1}{(t+1)^{0.75}}$, $\alpha(t)=\frac{1}{(t+1)^{0.5}}$ and  $\alpha(t)=\frac{1}{(t+1)^{0.25}}$, under which we plot in Figure \ref{fig:eg2} the trajectories of
$$e(t):=\big\|\sum\limits_{i=1}^4 \xb_i(t)/4-\yb^\ast_2\big\|.$$
By direct calculation, we find
$$e(t)=\MO(\frac{1}{t^{0.750}}),\ e(t)=\MO(\frac{1}{t^{0.492}}),\ e(t)=\MO(\frac{1}{t^{0.249}})$$
for $\alpha(t)=\frac{1}{(t+1)^{0.75}},\frac{1}{(t+1)^{0.5}},\frac{1}{(t+1)^{0.25}}$, respectively. These results validate the statement in Theorem \ref{thm:convergence_rate} for the step size $\alpha(t)=\MO(\frac{1}{t^\lambda}),\ \lambda\in(0,1)$.

\begin{figure}
\centering
\includegraphics[width=3.3in]{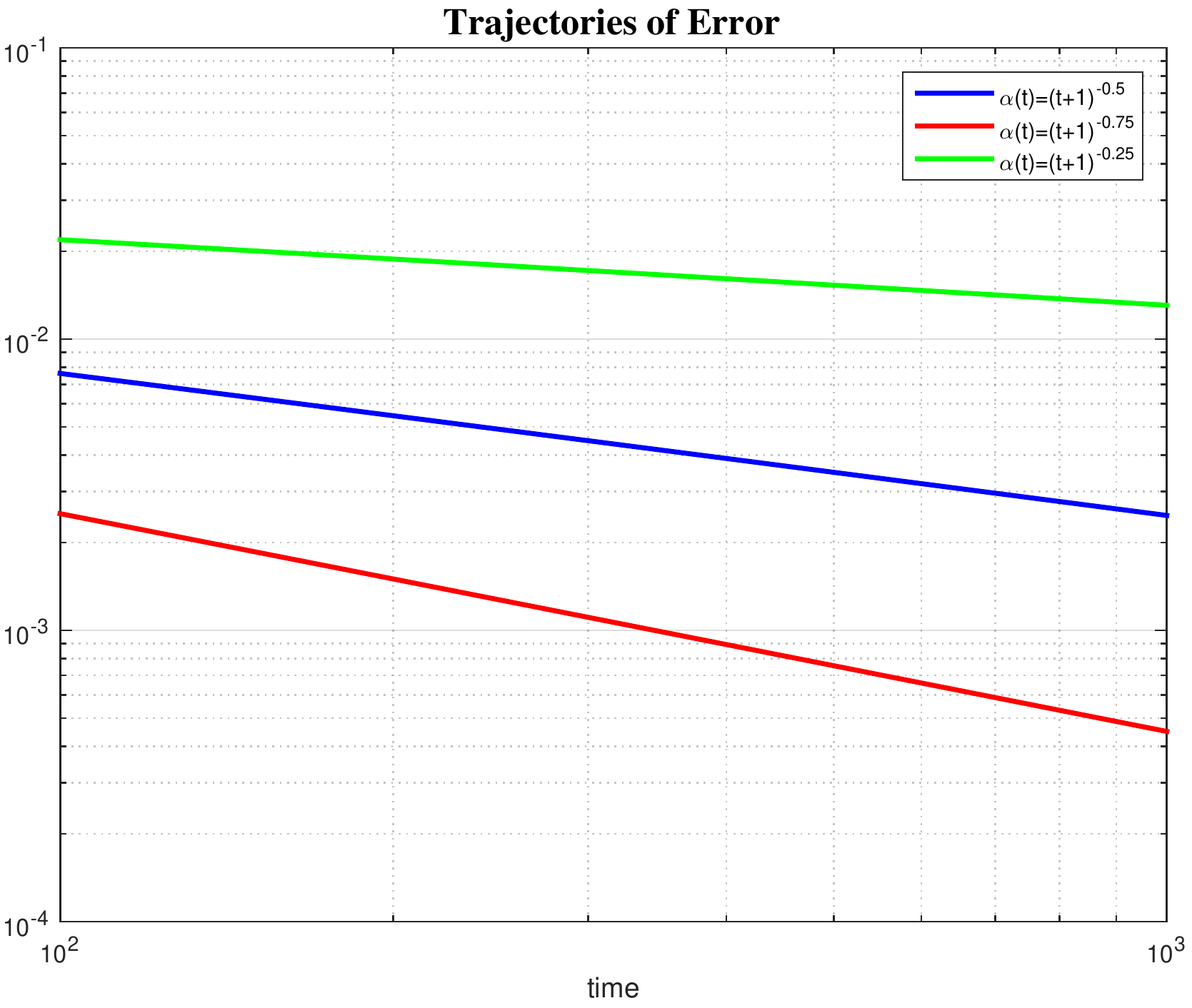}
\caption{The trajectories of $e(t):=\big\|\sum\limits_{i=1}^4 \xb_i(t)/4-\yb^\ast_2\big\|$ with $\yb_2^\ast=[-0.092\ 0.361]^\top$ for $\alpha(t)=\frac{1}{(t+1)^{0.75}},\ \alpha(t)=\frac{1}{(t+1)^{0.5}}$ and  $\alpha(t)=\frac{1}{(t+1)^{0.25}}$, respectively. The slopes are $-0.750,-0.492,-0.249$, respectively.}
\label{fig:eg2}
\end{figure}

}

\subsection{Switching Connected Graphs}
% \noindent{\bf Example 3.}
% Consider the following linear equation with respect to $\yb\in\R^2$:
% \[
% \begin{bmatrix}
% 4 & -2\\
% 2 & 3\\
% 3 & -2.5\\
% -1.5 & 2\\
% 1 & 4
% \end{bmatrix}
% \yb
% =
% \begin{bmatrix}
% 1\\
% 3\\
% 2\\
% 3\\
% -2
% \end{bmatrix}.
% \]
% As can be calculated, its unique least-squares solution is $\yb^\ast=[-0.3124\ 0.0677]^\top$. Let $\mathcal{Q}^\ast=\{\mathcal{G}_1,\mathcal{G}_2\}$ with $\mathcal{G}_1,\mathcal{G}_2$ as shown in Figure \ref{fig:l1l2} and $\mathcal{G}_{\sigma(t)}$ be given as following:
% $$\mathcal{G}_{\sigma(t)}=\left\{
% \begin{aligned}
% & \mathcal{G}_1,\ t\in\big [T k,T (k+1)\big ), k=0,2,4,\dots \\
% & \mathcal{G}_2,\ t\in\big [T k,T (k+1)\big ), k=1,3,5,\dots \\
% \end{aligned}
% \right.
% $$
% with $T=1$,
% i.e., the network switches between graph $\mathcal{G}_1$ and $\mathcal{G}_2$ periodically with period $T=1$. Set the initial value $\xb(0)=[3.5\ 4\ 5\ -4\ -4\ 3\ -2\ -3.4\ -5\ 4.5]^\top$. Let the flow (\ref{eq:x_ori}) do iteration over the switching network $\mathcal{G}_{\sigma(t)}$ with $K=100, \alpha(t)=(t+1)^{-1}$. It can be known by simple calculation that the conditions of Theorem \ref{thm:2} are met, in particular, $\rank(\Hb)=2$. Then we plot the trajectories $\xb_i[1](t)$, $\xb_i[2](t)$ with $i=1,2,3,4,5$ in Figure \ref{fig:convswitching_fullrank}. We can see that $\xb_i(t)$ for all $i$ converge to $\yb^\ast$, consistent with Theorem \ref{thm:2}.

\noindent{\bf Example 3.} Consider the following linear equation with respect to $\yb\in\R^2$:
\[
\begin{bmatrix}
4 & -2\\
2 & -1\\
3 & -1.5\\
-1.5 & 0.75\\
1 & -0.5
\end{bmatrix}
\yb
=
\begin{bmatrix}
1\\
3\\
2\\
3\\
-2
\end{bmatrix}.
\]
We can easily check that the conditions of Theorem \ref{thm:2} are satisfied, in particular, $\rank(\Hb)=1<2$, which means the linear equation has non-unique least-squares solutions. Let $\mathcal{Q}^\ast=\{\mathcal{G}_1,\mathcal{G}_2\}$ with $\mathcal{G}_1,\mathcal{G}_2$ as shown in Figure \ref{fig:l1l2} and $\mathcal{G}_{\sigma(t)}$ be given as following:
$$\mathcal{G}_{\sigma(t)}=\left\{
\begin{aligned}
& \mathcal{G}_1,\ t\in\big [T k,T (k+1)\big ), k=0,2,4,\dots \\
& \mathcal{G}_2,\ t\in\big [T k,T (k+1)\big ), k=1,3,5,\dots \\
\end{aligned}
\right.
$$
with $T=0.1$,
i.e., the network switches between graph $\mathcal{G}_1$ and $\mathcal{G}_2$ periodically with period $T=0.1$. Set the initial value $\xb(0)=[3.5\ 4\ 5\ -4\ -4\ 3\ -2\ -3.4\ -5\ 4.5]^\top$. Let the flow (\ref{eq:x_ori}) do iteration over the switching network $\mathcal{G}_{\sigma(t)}$ with $K=100, \alpha(t)=(t+1)^{-1}$. Then the trajectories of $\xb_i[1](t),\xb_i[2](t)$ with $i=1,2,3,4,5$ are plotted in blue in Figure \ref{fig:convswitching_notfullrank}, from which it can seen that $\xb_i(t)$ for all $i$ converge to $\hat{\yb}_1=[-0.1925\ 0.9737]^\top$. Next we reset the initial value as $\xb(0)=[-2\ 1.25\ -3\ 2\ 1\ 3\ 1.3\ 0.8\ -0.8\ 3.5]^\top$ and plot the states trajectories in red in Figure \ref{fig:convswitching_notfullrank}, and the new limit turns to be $\hat{\yb}_2=[-0.7491\ 2.0854]^\top$. Evidently, $\hat{\yb}_1$ and $\hat{\yb}_2$ are two different least-squares solutions and this simulation result is consistent with the claim of Theorem \ref{thm:2}. It also implies that, unsurprisingly, the initial values determine the value of the nonunique least-squares solution that the system state converges to.

% \begin{figure}
% \centering
% \includegraphics[width=1.5in]{L1.pdf}
% \caption{A constant, connected and undirected graph $\mathcal{G}_1$ considered in Example 3 and 4.}
% \label{fig:l1}
% \end{figure}

% \begin{figure}
% \centering
% \includegraphics[width=1.5in]{L2.pdf}
% \caption{A constant, connected and undirected graph $\mathcal{G}_2$ considered in Example 3 and 4.}
% \label{fig:l2}
% \end{figure}

\begin{figure} [htbp]
\hspace{5mm}
\subfigure[$\mathcal{G}_1$]{
\begin{minipage}{0.41\linewidth}
\centering
\includegraphics[width=3.5cm]{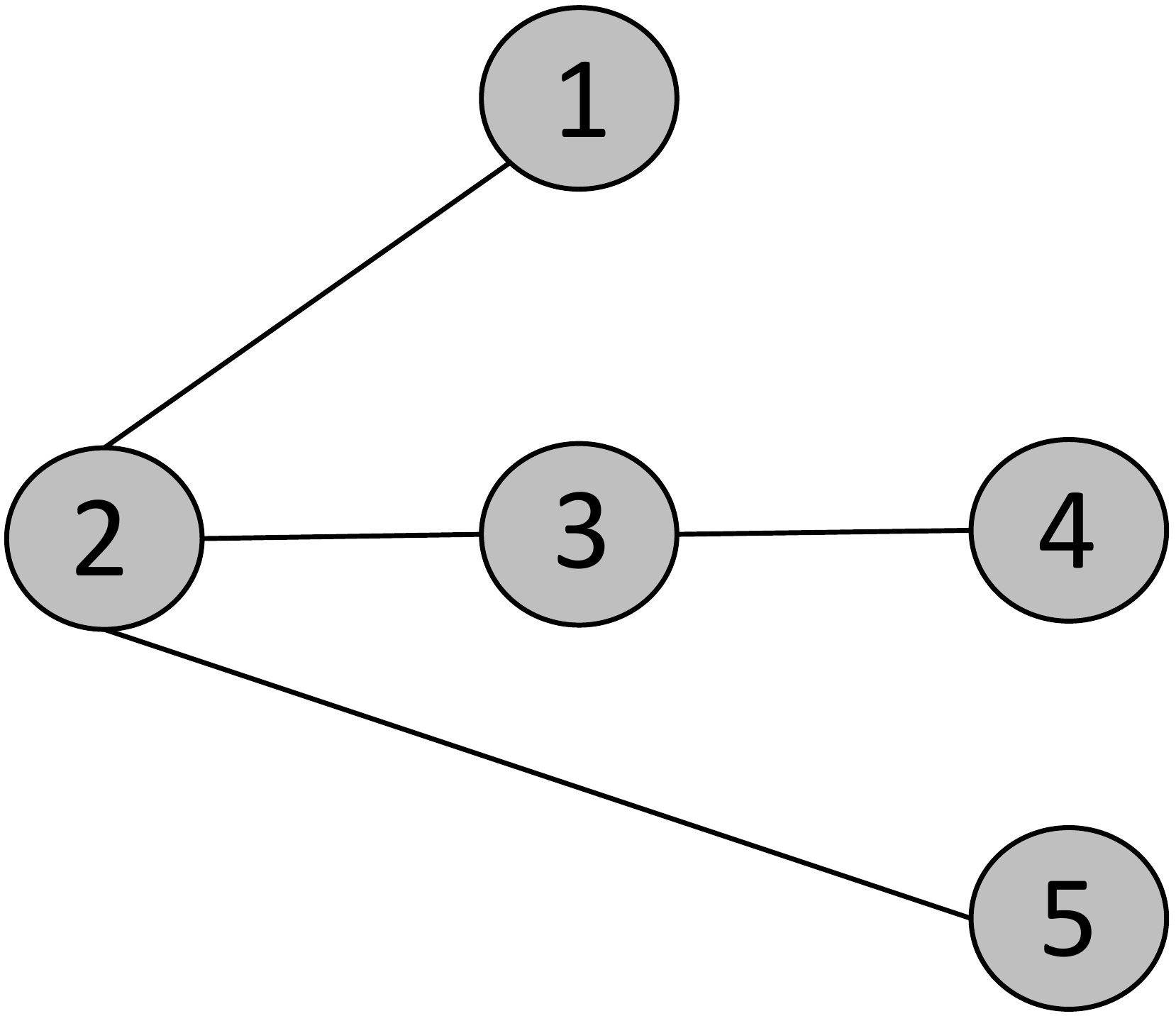}
\end{minipage}
}
\subfigure[$\mathcal{G}_2$]{
\begin{minipage}{0.44\linewidth}
\centering
\includegraphics[width=3.5cm]{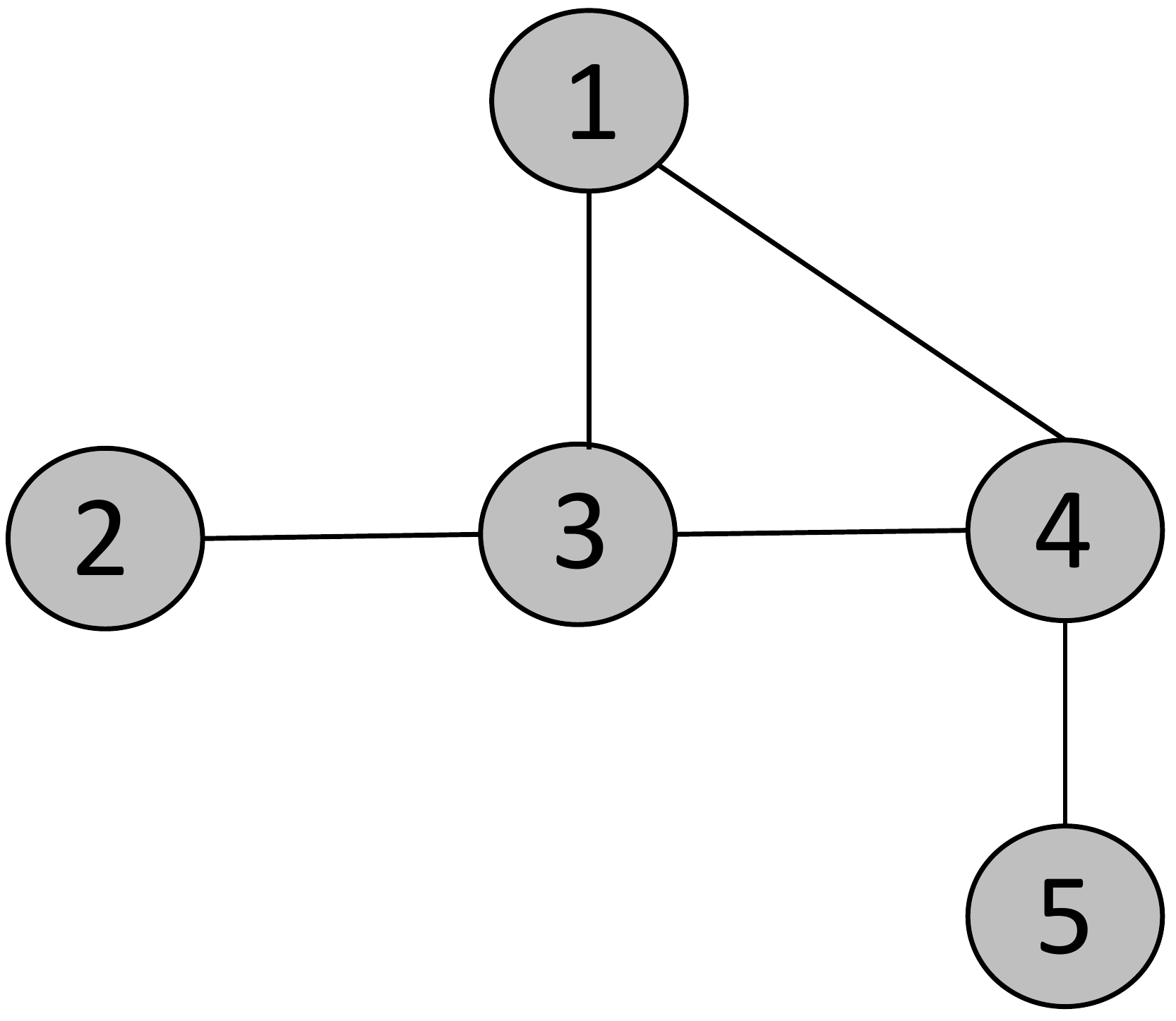}
\end{minipage}
}
\caption{Constant, connected and undirected graph $\mathcal{G}_1$, $\mathcal{G}_2$ considered in Example 3 and 4.}
\label{fig:l1l2}
\end{figure}

% \begin{figure}
% \centering
% \includegraphics[width=3.4in]{ConvSwitching_fullrank.pdf}
% \caption{The trajectories of $\xb_i[1](t), \xb_i[2](t)$ for $i=1,2,3,4,5$ given $K=100$, $\alpha(t)=(t+1)^{-1}$ obtained over a switching network. The result shows all $\xb_i(t)$ converge to $\yb^\ast=[-0.3124\ 0.0677]^\top$.}
% \label{fig:convswitching_fullrank}
% \end{figure}

\begin{figure}
\centering
\includegraphics[width=3.4in]{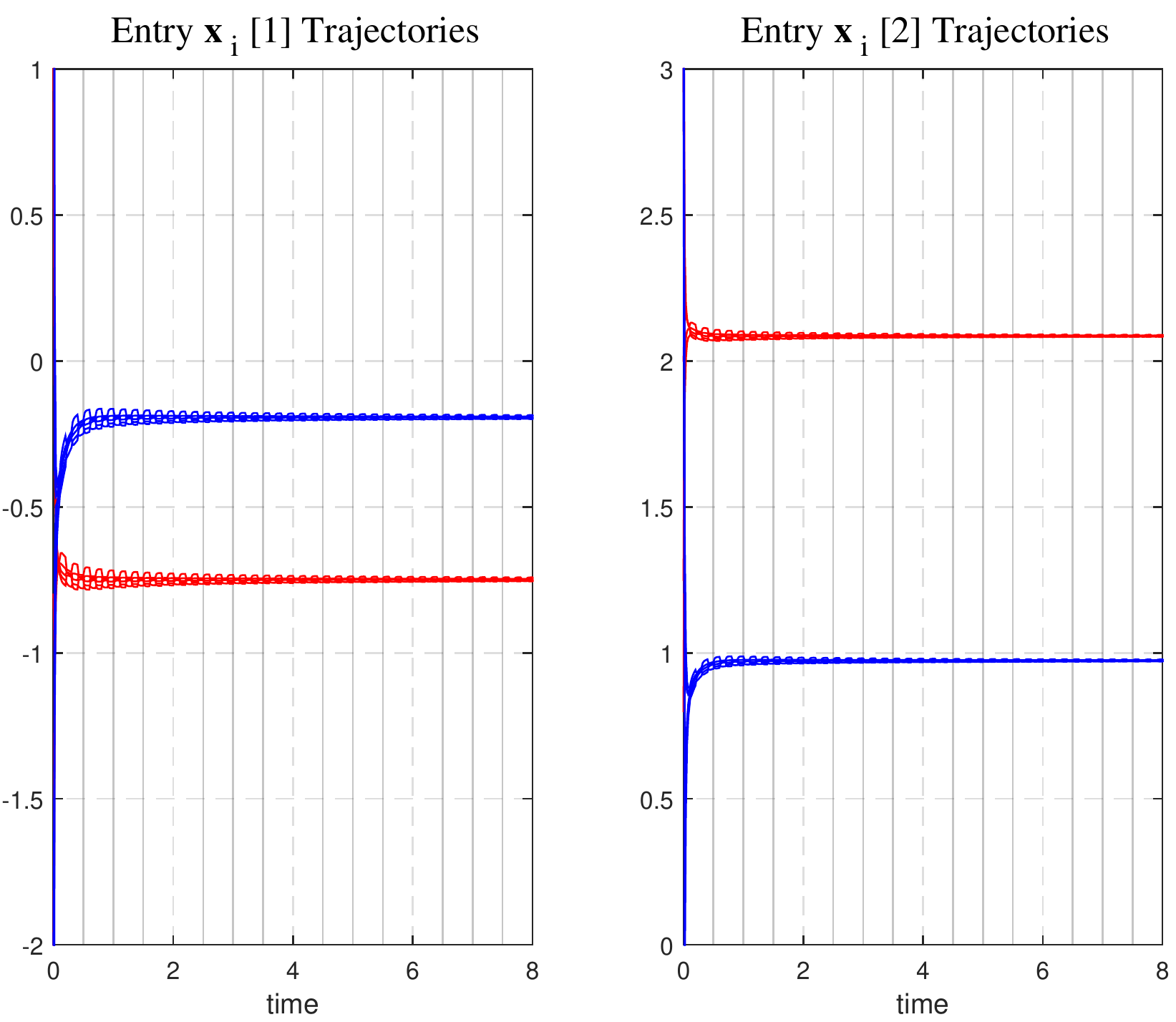}
\caption{The trajectories of the first component $\xb_i[1](t)$ and the second component $\xb_i[2](t)$ for $i=1,2,3,4,5$ given $K=100$, $\alpha(t)=(t+1)^{-1}$ obtained over a switching network with two different sets of initial values. As calculated, all $\xb_i(t)$ in blue converge to $\hat{\yb}_1=[-0.1925\ 0.9737]^\top$ and all $\xb_i(t)$ in red converge to $\hat{\yb}_2=[-0.7491\ 2.0854]^\top$, which are two different least-squares solutions.}
\label{fig:convswitching_notfullrank}
\end{figure}

% \begin{figure}
% \centering
% \includegraphics[width=3.4in]{ConvSwitching_notfullrank_2.pdf}
% \caption{The trajectories of the first component $\xb_i[1](t)$ and the second component $\xb_i[2](t)$ for $i=1,2,3,4,5$ given $K=100$, $\alpha(t)=(t+1)^{-1}$ obtained over a switching network. As calculated, all $\xb_i(t)$ converge to $[-0.7491,2.0584]^\top$, which is one of the least-squares solutions.}
% \label{fig:convswitching_notfullrank2}
% \end{figure}

\subsection{Switching Graphs with Joint Connectivity}
\noindent{\bf Example 4.}
Consider the following same linear equation as in Example 4. Let $\mathcal{G}_{\sigma(t)}$ be given as following:
$$\mathcal{G}_{\sigma(t)}=\left\{
\begin{aligned}
& \mathcal{G}_3,\ t\in\big [T k,T (k+1)\big ), k=0,2,4,\dots \\
& \mathcal{G}_4,\ t\in\big [T k,T (k+1)\big ), k=1,3,5,\dots \\
\end{aligned}
\right.
$$
with $\mathcal{G}_3$, $\mathcal{G}_4$ in Figure \ref{fig:l3l4}, $T=0.1$. We can see that neither $\mathcal{G}_3$ nor $\mathcal{G}_4$ is connected, but $\mathcal{G}_{\sigma(t)}$ is uniformly connected. Given the same $K,\alpha(t), \xb(0)=[3.5\ 4\ 5\ -4\ -4\ 3\ -2\ -3.4\ -5\ 4.5]^\top$ as Example 3. Let the flow (\ref{eq:x_ori}) do iteration over $\mathcal{G}_{\sigma(t)}$. Then we plot the trajectories of $\xb_i[1](t),\xb_i[2](t)$ for all $i$ in Figure \ref{fig:convswitching_fullrank_uniconnected}. It can be seen that $\xb_i(t)$ converge to $\yb^\ast=[-0.1925\ 0.9737]^\top$ for all $i$ when $\rank(\Hb)=m$, which is consistent with Theorem \ref{thm:3}. We can also verify the convergence for the case with $\rank(\Hb)<m$.

% \begin{figure}
% \centering
% \includegraphics[width=1.5in]{L3.pdf}
% \caption{A constant, connected and undirected graph $\mathcal{G}_3$ considered in Example 5}
% \label{fig:l3}
% \end{figure}

% \begin{figure}
% \centering
% \includegraphics[width=1.5in]{L4.pdf}
% \caption{A constant, connected and undirected graph $\mathcal{G}_4$ considered in Example 5}
% \label{fig:l4}
% \end{figure}

\begin{figure} [htbp]
\hspace{5mm}
\subfigure[$\mathcal{G}_3$]{
\begin{minipage}{0.41\linewidth}
\centering
\includegraphics[width=3.5cm]{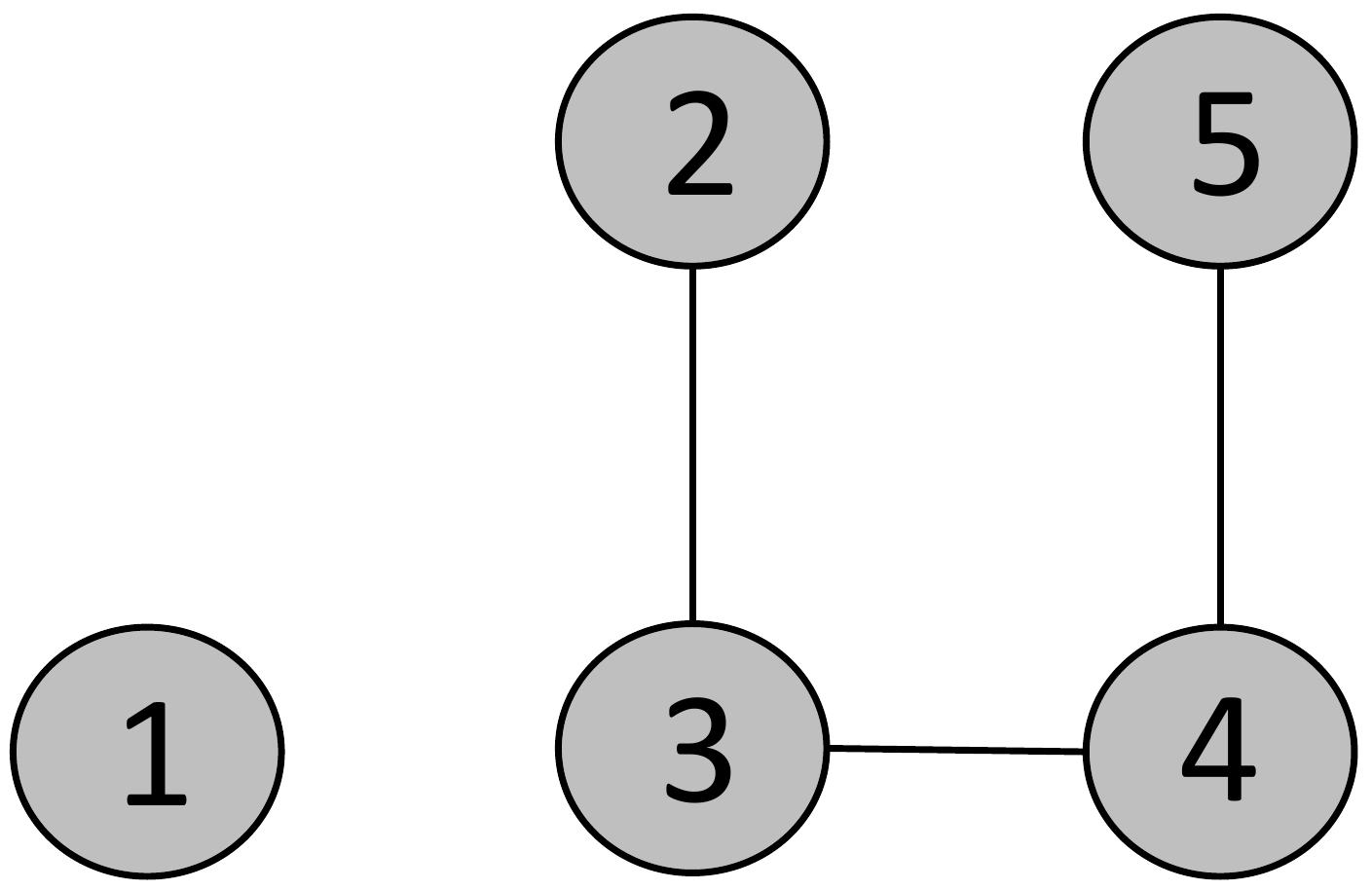}
\end{minipage}
}
\subfigure[$\mathcal{G}_4$]{
\begin{minipage}{0.44\linewidth}
\centering
\includegraphics[width=3.5cm]{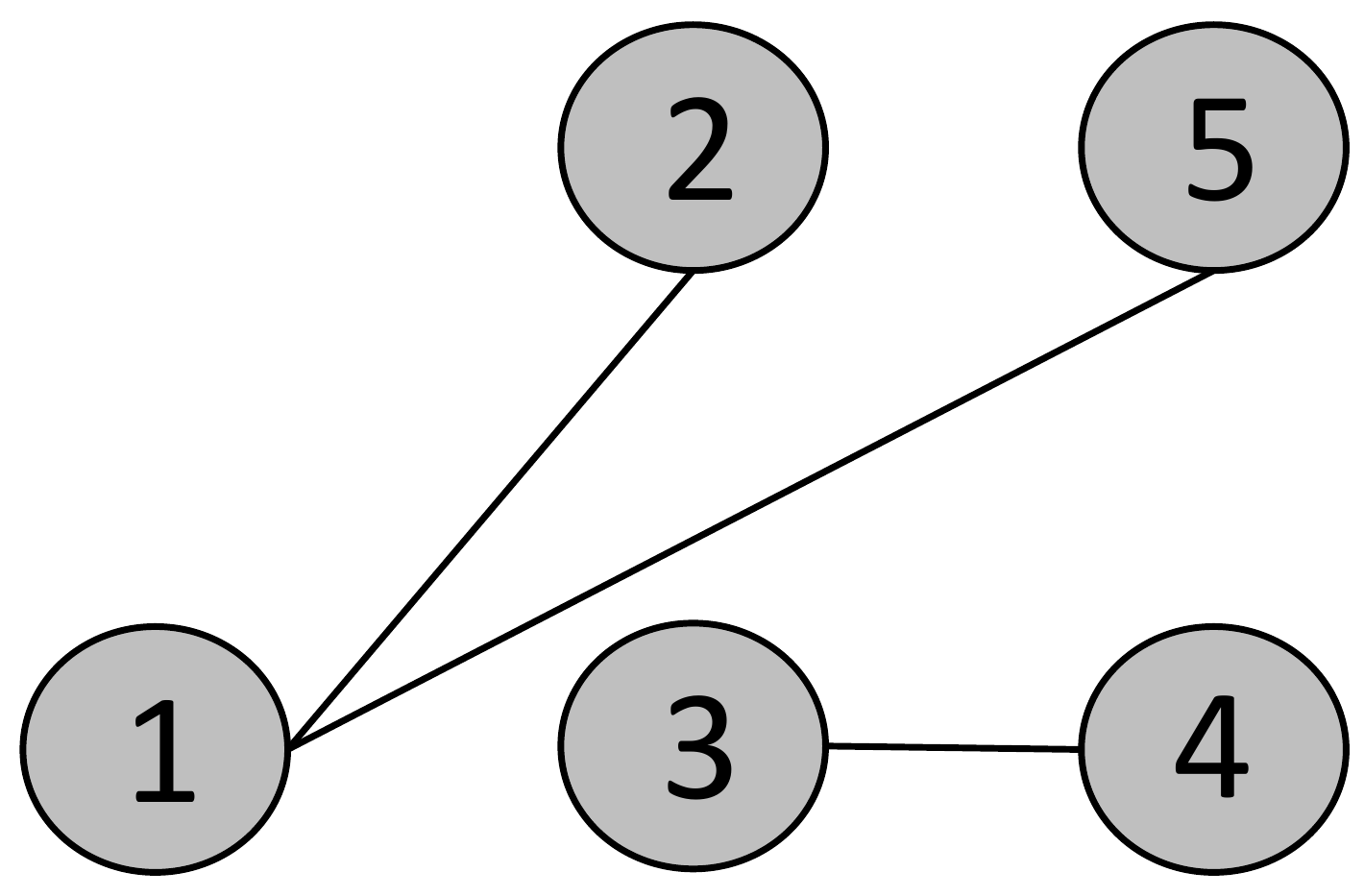}
\end{minipage}
}
\caption{Constant, connected and undirected graph $\mathcal{G}_3$, $\mathcal{G}_4$ considered in Example 5.}
\label{fig:l3l4}
\end{figure}

\begin{figure}
\centering
\includegraphics[width=3.4in]{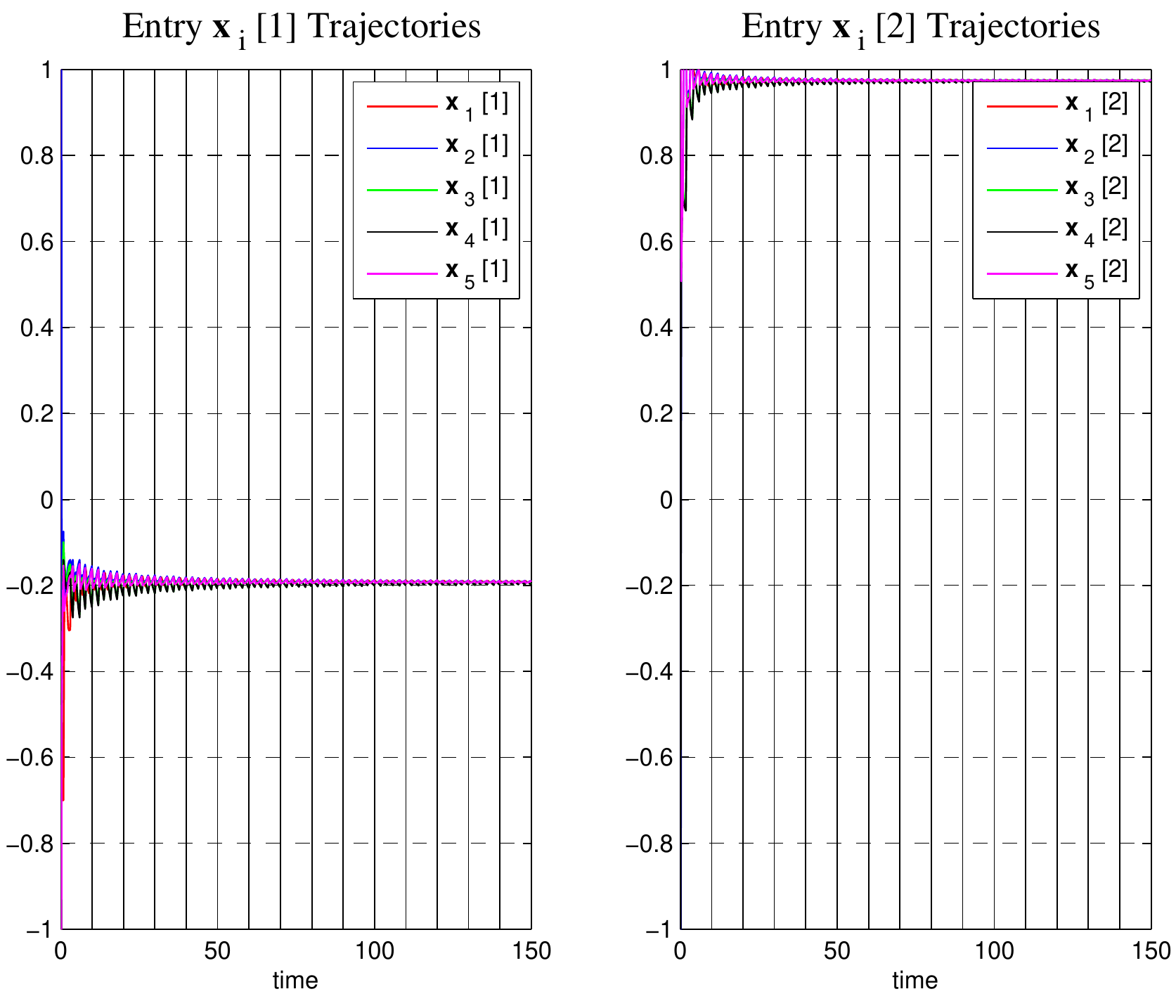}
\caption{The trajectories of the first component $\xb_i[1](t)$ and the second component $\xb_i[2](t)$ for $i=1,2,3,4,5$ given $K=100$, $\alpha(t)=(t+1)^{-1}$ obtained over a switching network with connected graph union. It can be seen that all $\xb_i(t)$ converge to $\yb^\ast=[-0.1925\ 0.9737]^\top$.}
\label{fig:convswitching_fullrank_uniconnected}
\end{figure}

\section{Conclusions}\label{sec:conclu}
In this paper, a first-order distributed {continuous-time} least-squares solver over networks was proposed. When the least-squares solution is unique, we proved the convergence results for fixed and connected graphs with an assumption of nonintegrable step size. {We also carefully analyzed the bound of convergence speed for two classes of step size choices, which provides guidance on the selection of step size to secure the fastest convergence speed.} By loosening the requirement for uniqueness of the least-squares solution and assuming square integrability on step size, we obtained convergence results for a constantly connected switching graph, and for uniformly jointly connected graphs under a boundedness assumption of system states. We also provided some numerical examples, in order to verify the results and illustrate the convergence speed. Potential future work includes proving the convergence over networks without instantaneous connectivity, studying the exact convergence rate, and finding out the convergence limit.

\end{document}